\documentclass[journal]{IEEEtran}
\usepackage{cite}

\usepackage{graphicx}
\ifCLASSINFOpdf
  
\else
\fi
\usepackage{amsmath,amssymb,amsthm}
\usepackage{algorithm}
\usepackage{algcompatible}

\usepackage[caption=false,font=footnotesize]{subfig}

\newtheorem{theorem}{Theorem}
\newtheorem{corollary}{Corollary}[theorem]
\newtheorem{lemma}[theorem]{Lemma}

\newtheorem{assumption}{Assumption}

\begin{document}
%
\title{Multiple Object Trajectory Estimation \\Using Backward Simulation}
%
%
%

\author{Yuxuan Xia, Lennart Svensson, {\'A}ngel F. Garc{\'\i}a-Fern{\'a}ndez, Jason L. Williams, \\Daniel Svensson, and Karl Granstr\"{o}m
\thanks{Y. Xia and L. Svensson are with the Department
of Electrical Engineering, Chalmers University of Technology, Gothenburg, Sweden. E-mail: firstname.lastname@chalmers.se. A. F. Garc{\'\i}a-Fern{\'a}ndez is with the Department of Electrical Engineering and Electronics, University of Liverpool, Liverpool, United Kingdom, and also with the ARIES research centre, Universidad Antonio de Nebrija, Madrid, Spain. J. L. Williams is with the Commonwealth Scientific and Industrial Research Organization, Brisbane, Australia. D. Svensson is with NVIDIA Corporation, Gothenburg, Sweden. K. Granstr\"{o}m is with Embark Trucks Inc., San Francisco, CA, USA.}
\thanks{The work of D. Svensson was done when he was with Zenseact AB, and the work of K. Granstr\"{o}m was done when he was with Chalmers University of Technology.}
}

\maketitle

\begin{abstract}
This paper presents a general solution for computing the multi-object posterior for sets of trajectories from a sequence of multi-object (unlabelled) filtering densities and a multi-object dynamic model. Importantly, the proposed solution opens an avenue of trajectory estimation possibilities for multi-object filters that do not explicitly estimate trajectories. In this paper, we first derive a general multi-trajectory backward smoothing equation based on random finite sets of trajectories. Then we show how to sample sets of trajectories using backward simulation for Poisson multi-Bernoulli filtering densities, and develop a tractable implementation based on ranked assignment. The performance of the resulting multi-trajectory particle smoothers is evaluated in a simulation study, and the results demonstrate that they have superior performance in comparison to several state-of-the-art multi-object filters and smoothers.
\end{abstract}


%
\IEEEpeerreviewmaketitle

\section{Introduction}

Multi-object tracking (MOT) refers to the problem of jointly estimating the number of objects and their trajectories from noisy sensor measurements \cite{bar2004estimation,challa2011fundamentals,meyer2018message,streit2021analytic}. Vector-type MOT methods, e.g., the joint probabilistic data association filter (JPDAF) \cite{bar2009probabilistic} and the multiple hypothesis tracker (MHT) \cite{blackman2004multiple,chong2019forty}, describe the multi-object states and measurements by random vectors; they explicitly estimate trajectories by linking a state estimate with a previous state estimate or declare the appearance of a new object. However, for MOT methods based on sets representation of the multi-object states, e.g., \cite{mahler2007statistical,mahler2003multitarget,mahler2007phd,pmbmpoint}, sequences of object states at consecutive time steps cannot be easily constructed.

For these MOT methods, one approach to estimating trajectories is to add a unique label to each single-object state such that each object can be identified over time \cite{glmbconjugateprior,garcia2013two,garci2014bayesian,streit2018analytic,aoki2016labeling}. This track labelling procedure may work well in some cases, but it often becomes problematic in challenging scenarios, for example, where initially well-separated objects move in close proximity with each other and thereafter separate again \cite{aoki2016labeling,garcia2019multiple,granstrom2018poisson}. A more advantageous approach to estimating trajectories for filters based on random finite sets (RFSs) \cite{mahler2007statistical} is to compute the multi-object posterior on sets of trajectories \cite{garcia2019multiple}, which captures all the information about the trajectories. This has led to the development of a variety of multi-object trackers: the trajectory probability hypothesis density (PHD) filter \cite{garcia2019trajectory}, the trajectory cardinality PHD filter \cite{garcia2019trajectory}, the trajectory Poisson multi-Bernoulli mixture (PMBM) filter \cite{granstrom2018poisson}, the trajectory multi-Bernoulli mixture (MBM) filter \cite{xia2019multi}, and their approximations the trajectory Poisson multi-Bernoulli (PMB) filter \cite{garcia2020trajectory} and the trajectory multi-Bernoulli (MB) filter \cite{garcia2020trajectory3}. Note that for RFS-based filters with MB birth, the multi-object posterior may be labelled to consider sets of labelled trajectories \cite{garcia2019multiple,xia2019multi,vo2019multi,garcia2020trajectory3}.

\begin{figure}[!t]
  \centering
  \includegraphics[width=\columnwidth]{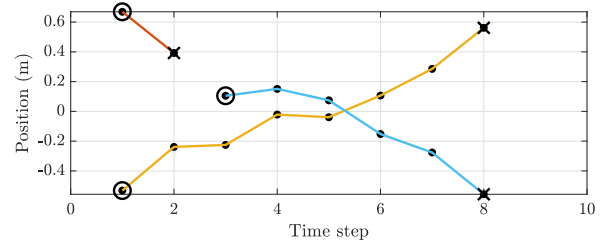}
  \caption{Illustration of a one-dimensional example where the multi-object filtering density at each time step is a Dirac delta, whose corresponding object states are shown as black dots. The trajectory building problem in this case reduces to how to connect the dots at consecutive time steps. The solid lines represent one of the many ways to construct a set of trajectories where start/end positions of trajectories are marked by circles/crosses. This example will be further elaborated in Section \ref{sec_example}.}
  \label{fig_intro_example}
\end{figure}

Smoothing for state-space models considers the estimation of object states of interest conditioned on the complete measurement sequence \cite{sarkka2013bayesian}. Therefore, smoothing may provide significantly better object state estimation performance than filtering, by refining earlier object state estimates. Solutions to single-object smoothing in clutter mainly include the Gaussian sum smoother \cite{kitagawa1994two,lee2015smoothing,rahmathullah2014merging,rahmathullah2014two,balenzuela2018accurate}, the (integrated) probabilistic data association smoother \cite{chakravorty2006augmented,song2012smoothing,muvsicki2013smoothing} and the Bernoulli smoother \cite{clark2009joint,vo2011bernoulli,vo2011closed}. For vector-type MOT methods with Gaussian filtering, smoothed trajectory estimates are typically obtained by applying a Rauch-Tung-Striebel (RTS) smoother \cite{rauch1965maximum} on sequences of single-object filtering densities \cite{mahalanabis1990improved,koch2000fixed,chen2001tracking}. It is also possible to consider a batch solution for estimating trajectories using expectation-maximisation \cite{rahmathullah2016batch}. As a comparison, trajectory filters \cite{garcia2019multiple,granstrom2018poisson,garcia2019trajectory,garcia2020trajectory,xia2019multi,garcia2020trajectory3,vo2019multi} recursively compute the posterior of sets of trajectories as new observations arrive, by performing smoothing-while-filtering \cite{briers2010smoothing}, and therefore the initiation and termination of estimated trajectories can also be improved.

Nevertheless, there are several MOT methods in the literature, e.g., the set JPDAF \cite{sjpda} and the variational PMB filter \cite{variational}, that can efficiently estimate the multi-object states, but that cannot easily produce smoothed trajectory estimates in a principled manner\footnote{A heuristic track-to-target management scheme was presented in \cite{svensson2011multitarget} by making use of the permutation probabilities of state vectors.}. Then an interesting research question arises: ``{\em Can we leverage filters that do not keep trajectory information to compute the posterior density of sets of trajectories?}''. A one-dimensional example of such trajectory building is illustrated in Fig. \ref{fig_intro_example}. Also note that, even for methods that do retain implicit trajectory information, in complex scenarios, approximations made for computational tractability (such as pruning) may cause loss of information that could be recovered with the help of information from later time steps.

In this paper, we show that the exact multi-object posterior of sets of trajectories can be obtained from a sequence of multi-object (unlabelled) filtering densities using the multi-object dynamic model. Specifically, we present a general multi-trajectory backward smoothing equation based on sets of trajectories. The proposed solution has important advantages over multi-object forward-backward smoothers that only compute the marginal multi-object smoothing densities at each time step \cite{nandakumaran2007improved,clark2010first,nadarajah2011multitarget,mahler2012forward,feng2016adaptive,li2016multi,nagappa2017tractable,streit2017interval}, which, even if labelled, may not be enough to provide meaningful trajectory information \cite[Example 2]{garcia2019multiple}. Moreover, the proposed forward-backward smoother does not specify the form of the multi-object filtering densities. This is in contrast to multi-object forward-backward smoothers based on labelled RFSs \cite{beard2016generalised,liu2019computationally,liang2020multitarget}, which cannot incorporate the Poisson birth model and require that the multi-object filtering densities must be labelled. The outcome of this work is a method for efficiently sampling multi-object trajectories from the posterior distribution of sets of trajectories, based on operations involving only the single time step multi-object state distributions constructed during forward filtering. This has important applications to offline trajectory analytics, e.g., extracting trajectory estimates that can be viewed as ground truth for the development and verification of perception modules in autonomous driving.

A preliminary version of this work was presented in \cite{xia2020backward}. This paper is a significant extension of that work, and contains the following contributions:
\begin{enumerate}
  \item We derive a general multi-trajectory backward smoothing equation based on sets of trajectories.
  \item We propose a multi-trajectory particle smoother using backward simulation \cite{lindsten2013backward} for PMB filtering densities. In particular, this is a method for doing inference (drawing samples) in the multi-object trajectory space while only ever manipulating single time step, single object marginal distributions, which is only possible because of the properties of the multi-trajectory dynamic model that is stated in Section \ref{dynamic_model}. 
  \item We present a tractable implementation of the proposed smoother based on a linear-Gaussian dynamic model and ranked assignment.
  \item We compare the proposed algorithm to several state-of-the-art algorithms \cite{granstrom2018poisson,nguyen2019glmb,vo2019multi} in a simulation study, and the results demonstrate that the proposed methods have superior performance.
\end{enumerate}

The rest of the paper is organised as follows. The background on sets of trajectories, multi-object dynamic models and PMB filtering densities is introduced in Section II. The backward smoothing equation for sets of trajectories is presented in Section III. A multi-trajectory particle smoother for PMB filtering densities and its tractable implementation are given in Section IV and Section V, respectively. The simulation results are shown in Section VI, and the conclusions are drawn in Section VII.

\section{Background}

The notation in this paper is defined following the convention in \cite{garcia2019multiple,granstrom2019poisson}. For a generic space $D$, the set of finite subsets of $D$ is denoted by ${\cal F}(D)$, and the cardinality of a set $A\in {\cal F}(D)$ is $|A|$. The sequence of ordered positive integers $(\alpha,\alpha+1,\dots,\gamma-1,\gamma)$ is denoted by $\alpha:\gamma$, and the set that includes all the permutations of $1:n$ is denoted by $\Gamma_n$. We use $\uplus$ to denote union of sets that are mutually disjoint, $\langle f,g \rangle$ to denote the inner product $\int f(x)g(x) dx$, and the multi-object exponential $f^A$, for some real-valued function $f$, to denote the product $\prod_{x\in A}f(x)$ with $f^{\emptyset} = 1$ by convention. In addition, we use $\delta_x(\cdot)$ and $\delta_x[\cdot]$ to represent the Dirac and Kronecker delta functions centred at $x$, respectively. 

\subsection{State variables}

The single-object state is described by a vector $x \in \mathbb{R}^{n_x}$, typically containing the kinematic information about the object (e.g., position and velocity). A trajectory is represented as a variable $X = (t,x^{1:\nu})$ where $t$ is the initial time step of the trajectory, $\nu$ is its length, and $x^{1:\nu} = (x^1,\dots,x^\nu)$ denotes a finite sequence of length $\nu$ that contains the object states at time steps $t:t+\nu-1$. For two time steps $\alpha$ and $\gamma$, $\alpha \leq \gamma$, a trajectory $(t,x^{1:\nu})$ in the time interval $\alpha:\gamma$ existing from time step $t$ to $t+\nu-1$ satisfies that $\alpha \leq t \leq t+\nu-1 \leq \gamma$, and the variable $(t,\nu)$ hence belongs to the set $I_{(\alpha,\gamma)} = \{(t,\nu):\alpha \leq t \leq \gamma~\text{and}~ 1\leq \nu \leq \gamma - t +1\}$. A single trajectory in the time interval $\alpha:\gamma$ therefore belongs to the space $T_{(\alpha,\gamma)} = \uplus_{(t,\nu)\in I_{(\alpha,\gamma)}}\{t\}\times \mathbb{R}^{\nu n_x}$. We note that trajectory $X$ is a combination of discrete and continuous states. Such a hybrid state is not uncommon in MOT: a typical example is the interacting multiple model \cite{blom1988interacting}. 

A set ${\bf x} \in {\cal F}(\mathbb{R}^{n_x})$ of single-object states is a finite subset of $\mathbb{R}^{n_x}$, and a set ${\bf X}_{\alpha:\gamma} \in {\cal F}(T_{(\alpha,\gamma)})$ of trajectories is a finite subset of $T_{(\alpha,\gamma)}$. The subset of trajectories in ${\bf X}_{\alpha:\gamma}$ that were alive at time step $\eta$ where $\alpha \leq \eta \leq \gamma$ is denoted by 
\begin{equation*}
  {\bf X}_{\alpha:\gamma}^{\eta} = \left\{ \left(t,x^{1:\nu} \right)\in {\bf X}_{\alpha:\gamma}: t \leq \eta \leq t+\nu-1 \right\}.
\end{equation*}
Given a set of trajectories ${\bf X}_{\alpha:\gamma}$, we denote the resulting set of trajectories in the time interval $\eta:\zeta$ by
\begin{multline*}
  {\bf X}_{\eta:\zeta} = \Big\{\left(\epsilon, x^{\epsilon-\eta+1:\iota-\epsilon+1}\right) : \left(t, x^{1:\nu}\right) \in {\bf X}_{\alpha:\gamma}, \\ \epsilon = \max(\eta, t), \iota = \min(\zeta, t+\nu-1), \epsilon \leq \iota \Big\}.
\end{multline*}
Note that ${\bf X}_{\eta :\zeta }$ depends on  ${\bf X}_{\alpha :\gamma }$, but we keep this relation implicit for notational clarity. An illustrative example of ${\bf X}_{\alpha:\gamma}$, ${\bf X}_{\alpha:\gamma}^{\eta}$ and ${\bf X}_{\eta:\zeta}$ is given in Fig. \ref{fig_notation_example}.

Given a single-object trajectory $X = (t,x^{1:\nu})$, the set of object states at time step $k$ is 
\begin{equation*}
  \tau^k(X) = \begin{cases}
    \left\{x^{k+1-t}\right\}, & t \leq k \leq t + \nu - 1 \\
    \emptyset, & \text{otherwise} 
  \end{cases}
\end{equation*}
and given a set ${\bf X}_{\alpha:\gamma}$ of trajectories, the set of object states at time step $k$ is $\tau^k({\bf X}_{\alpha:\gamma}) = \bigcup_{X\in{\bf X}_{\alpha:\gamma}}\tau^k(X)$. 

\begin{figure*}[!t]
  \centering
  \subfloat[\label{ex1}]
  {\includegraphics[width=0.666\columnwidth]{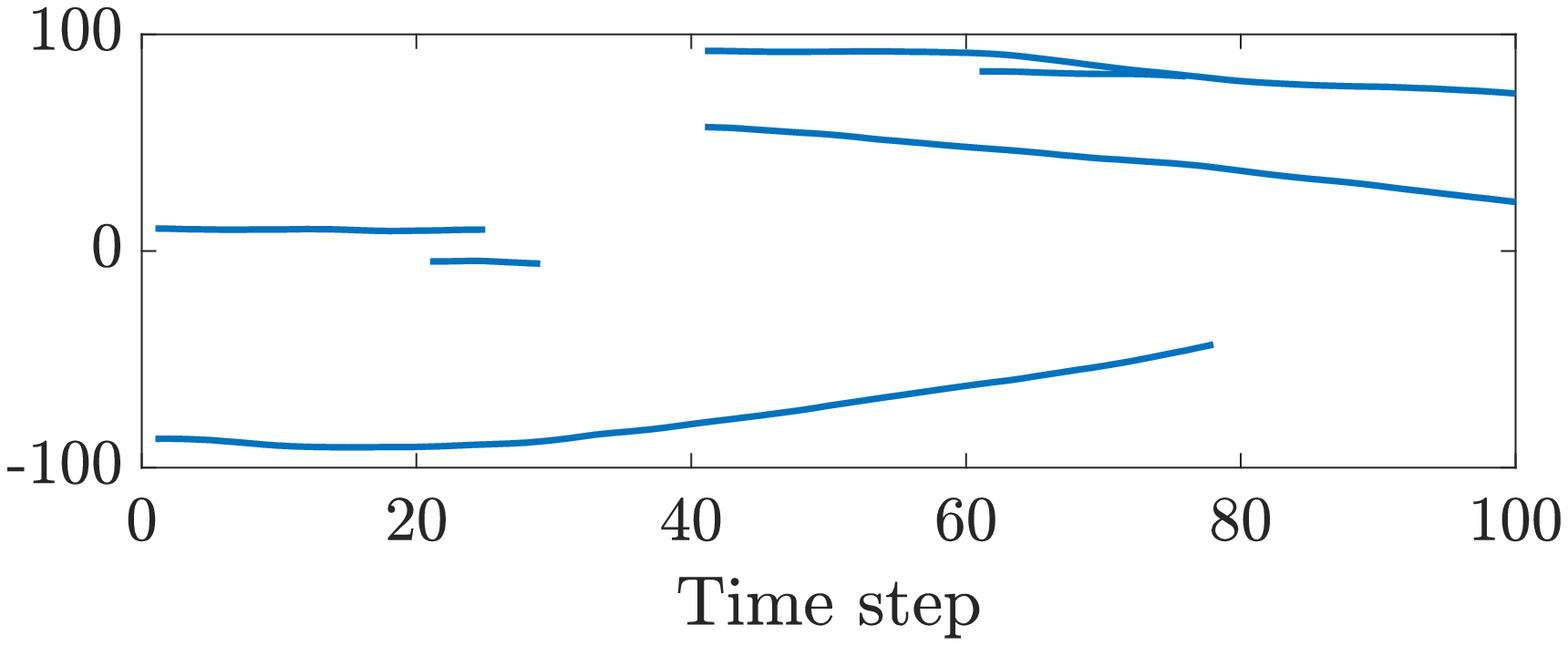}}
  \subfloat[\label{ex2}]
  {\includegraphics[width=0.666\columnwidth]{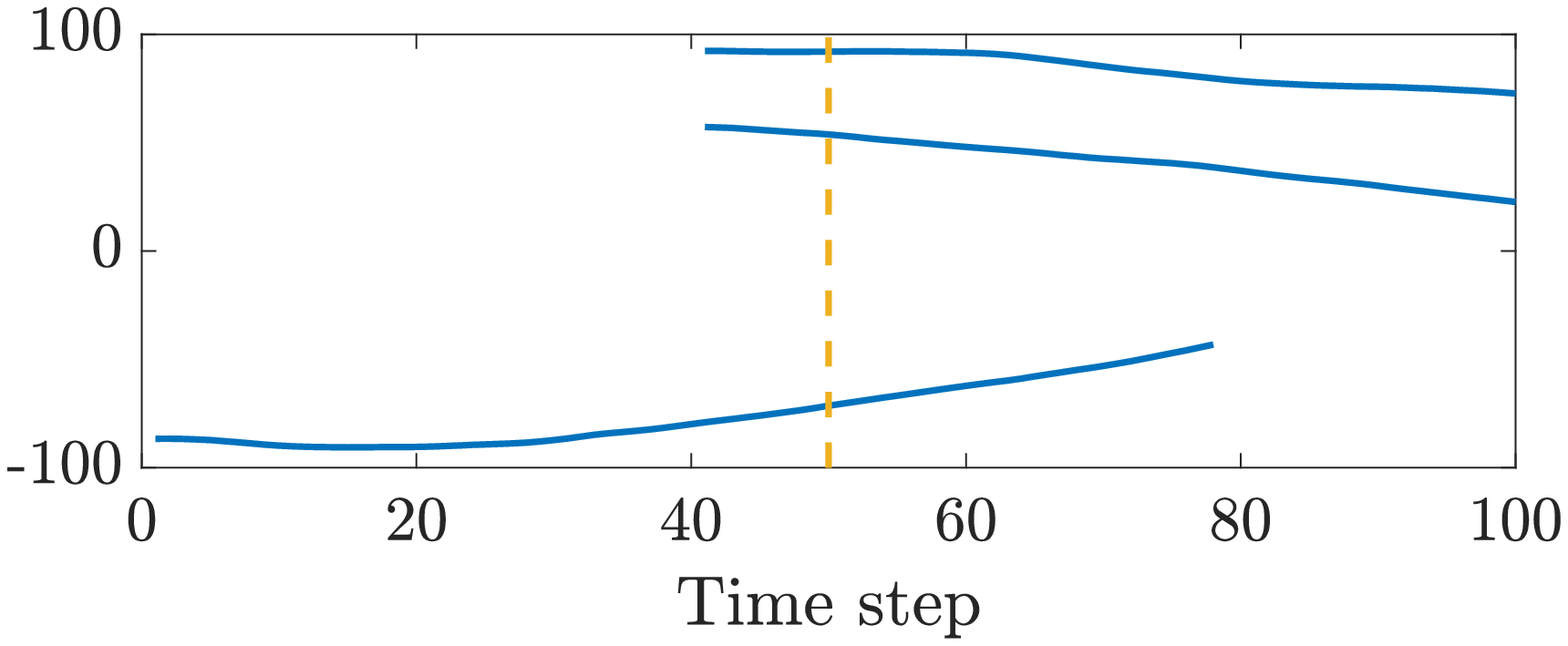}}
  \subfloat[\label{ex3}]
  {\includegraphics[width=0.666\columnwidth]{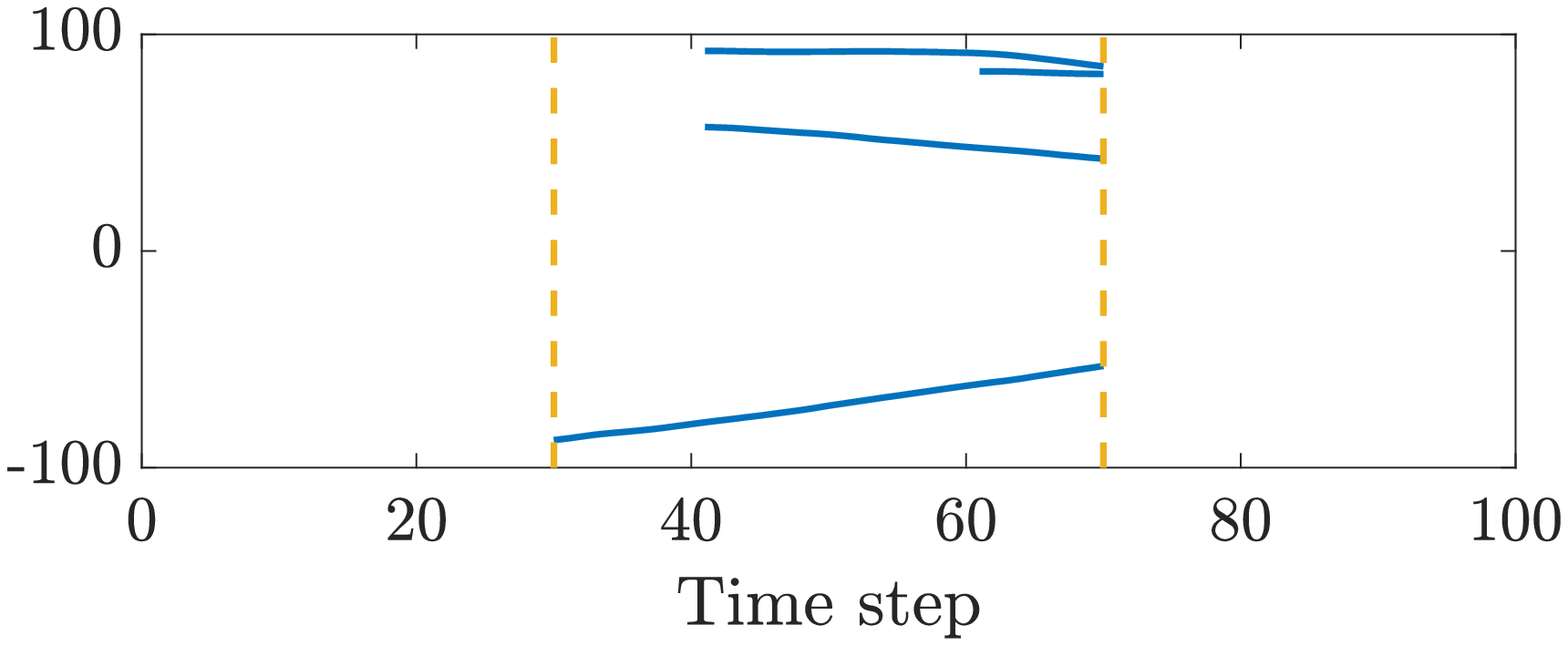}}
  \caption{One-dimensional examples of sets of trajectories: (a) a set ${\bf X}_{1:100}$ of six trajectories in the time interval $1:100$; (b) the set ${\bf X}^{50}_{1:100}$ of three trajectories in the time interval $1:100$ that were alive at time step $50$; (c) the resulting set ${\bf X}_{30:70}$ of four trajectories in the time interval $30:70$.}
  \label{fig_notation_example}
\end{figure*}

\subsection{Densities and integrals}

Given a real-valued function $\pi(\cdot)$ on the single trajectory space $T_{(\alpha,\gamma)}$, its integral is \cite{garcia2019multiple}
\begin{equation}
  \int \pi(X) d X = \sum_{(t,\nu)\in I_{(\alpha,\gamma)}} \int \pi\left(t,x^{1:\nu}\right) d x^{1:\nu},
\end{equation}
which goes through all possible start times, lengths and object states of trajectory $X \in T_{(\alpha,\gamma)}$. The single trajectory space is locally compact, Hausdorff and second-countable \cite{garcia2019multiple}, and therefore one can perform inference on sets of finite number of trajectories with finite length \cite{mahler2007statistical}. 


Given a real-valued function $\pi(\cdot)$ on the space ${\cal F}(T_{(\alpha,\gamma)})$ of finite sets of trajectories, its set integral is \cite{garcia2019multiple}
\begin{equation}
  \int \pi({\bf X}) \delta {\bf X} = \pi(\emptyset) + \sum_{n=1}^{\infty}\frac{1}{n!}\int \pi\left(\left\{X_1,\dots,X_n\right\}\right) d X_{1:n}
\end{equation}
where $X_{1:n} = (X_1,\dots,X_n)$. A function $\pi(\cdot)$ on the space ${\cal F}(T_{(\alpha,\gamma)})$ is a multi-trajectory density if $\pi(\cdot) \geq 0$ and its set integral is one. 


The multi-object state density at time step $k^\prime$ and the multi-trajectory density in the time interval $\alpha:\gamma$, both conditioned on the sequence of sets of measurements ${\bf z}_{1:k} = ({\bf z}_1,\dots,{\bf z}_k)$ received up to and including time step $k$, are denoted by $f_{k^\prime|k}(\cdot)$ and $\pi_{\alpha:\gamma|k}(\cdot)$, respectively. Given a set ${\bf x}_k$ of object states at time step $k$, the set of trajectories in the time interval $k:k$ is 
\begin{equation*}
  {\bf X}_{k:k} = \left\{X = \left(k,x^1\right): x^1 \in {\bf x}_k\right\}
\end{equation*}
where trajectory $X=(t,x^{1:\nu})\in {\bf X}_{k:k}$ has start time $t=k$ and length $\nu=1$ with probability one. Therefore, it holds that the multi-trajectory density $\pi_{k^\prime:k^\prime|k}({\bf X}_{k^\prime:k^\prime})$ takes the same value as the multi-object state density $f_{k^\prime|k}({\bf x})$ when $\tau^{k^\prime}({\bf X}_{k^\prime:k^\prime}) = {\bf x}$ where $k^\prime \in \{k,k+1\}$.


A trajectory Poisson point process (PPP) has density 
\begin{equation}
  \label{eq_trajectory_ppp}
  \pi\left({\bf X}_{\alpha:\gamma}\right) = e^{-\left\langle \lambda,1 \right\rangle} \left[ \lambda(\cdot) \right]^{{\bf X}_{\alpha:\gamma}}
\end{equation}
where $\lambda(\cdot)$ is the Poisson intensity, and a trajectory Bernoulli process has density 
\begin{equation}
  \label{eq_Bernoulli}
  \pi\left({\bf X}_{\alpha:\gamma}\right) = \begin{cases}
    1 - r, & {\bf X}_{\alpha:\gamma} = \emptyset \\
    rp(X), & {\bf X}_{\alpha:\gamma} = \{X\} \\
    0, & \text{otherwise}
  \end{cases}
\end{equation}
where $p(\cdot)$ is a single-trajectory density and $r$ is the probability of existence. A trajectory MB process is the union of $n\geq 1$ independent trajectory Bernoulli components and its density is given by the convolution formula for RFSs \cite[Sec. 11.5.3]{mahler2007statistical}
\begin{equation}\label{eq_mbconv}
  \pi\left({\bf X}_{\alpha:\gamma}\right) = \sum_{\uplus^{n}_{l=1} {\bf X}^l = {\bf X}_{\alpha:\gamma}} \prod_{i=1}^n \pi^{i}\left({\bf X}^i\right)
\end{equation}
where $\pi^{i}(\cdot)$ is the density of the $i$-th Bernoulli component, and the sum in \eqref{eq_mbconv} goes through all disjoint and possibly empty subsets ${\bf X}^1,\dots,{\bf X}^n$ such that ${\bf X}^1 \cup\cdots\cup {\bf X}^n = {\bf X}_{\alpha:\gamma}$.

Given two single-object trajectories $X = (t,x^{1:\nu})$ and $Y = (t^\prime,y^{1:\nu^\prime})$, the trajectory Dirac delta function is defined as 
\begin{equation*}
  \delta_{Y}(X) = \delta_{t^\prime}[t]\delta_{\nu^\prime}[\nu]\delta_{y^{1:\nu^\prime}}\left(x^{1:\nu}\right),
\end{equation*}
and the multi-trajectory Dirac delta function centred at ${\bf Y}$ is defined as \cite[Sec. 11.3.4.3]{mahler2007statistical}
\begin{equation*}
  \delta_{{\bf Y}}({\bf X}) = \begin{cases}
    0, & |{\bf X}| \neq |{\bf Y}|\\
    1, & {\bf X} = {\bf Y} = \emptyset\\
    \sum_{\sigma\in\Gamma_n}\prod_{i=1}^n\delta_{Y_{\sigma_i}}(X_i) & \begin{cases}
      {\bf X} = \{X_i\}_{i=1}^n \\
      {\bf Y} = \{Y_i\}_{i=1}^n
    \end{cases}.
  \end{cases}
\end{equation*}

\subsection{Multi-trajectory dynamic model}
\label{dynamic_model}
The conventional multi-object dynamic model described in \cite{mahler2007statistical} is considered. Given the current multi-object state ${\bf x}$, each object $x \in {\bf x}$ survives with probability $p^S(x)$, and moves to a new state with a Markovian transition density $g(\cdot|x)$, or dies with probability $1-p^S(\cdot)$. The multi-object state at the next time step is the union of the surviving objects and new objects, which are born independently of the rest. The newborn objects are typically modelled as a PPP or an MB process with multi-object state density $\beta(\cdot)$.

The above conventional multi-object dynamic model results in the following dynamic model for the set of all trajectories that have existed up to the current time step, which will be required in backward smoothing for sets of trajectories. Given a set ${\bf X}_{1:k}$ of all trajectories in the time interval $1:k$, each trajectory $X = (t,x^{1:\nu}) \in {\bf X}_{1:k}$ ``survives'' with probability one, $p^S(X) = 1$, and moves to a new state according to \cite{garcia2019multiple}
\begin{align}
  \label{eq_single_trajectory_transition}
  g^{k+1}\left(t^\prime,y^{1:\nu^\prime} | X\right) =& \left|\tau^{k}(X)\right| \left[ \left(1-p^S\left(x^\nu\right)\right)\delta_X\left(t^\prime,y^{1:\nu^\prime}\right)\right.\nonumber\\
  &+ \left. p^S\left(x^\nu\right) g\left(y^{\nu^\prime} | x^{\nu}\right)\delta_X\left(t^\prime,y^{1:\nu^\prime-1}\right) \right]\nonumber\\
  &+ \left(1 - \left|\tau^{k}(X)\right|\right)\delta_X\left(t^\prime,y^{1:\nu^\prime}\right).
\end{align}
That is, if the object underlying trajectory $X$ has died before time step $k$, the trajectory remains unaltered with probability one. If trajectory $X$ exists at time step $k$, it remains unaltered with probability $1-p^S(x^\nu)$, or the new final object state $y^{\nu^\prime}$ is generated according to the single-object transition density with probability $p^S(x^\nu)$. The set ${\bf X}_{1:k+1}$ of trajectories in the time interval $1:k+1$ is the union of the dead trajectories, surviving trajectories and new trajectories where each new trajectory $(t,x^{1:\nu})$ has deterministic start time $t=k+1$, length $\nu=1$, and multi-object state is distributed as $\beta(\cdot)$.

\subsection{PMB filtering densities}

RFS-based Bayes filters propagate the multi-object posterior density of ${\bf x}_k$ in time via the prediction and update steps:
\begin{align}
  f_{k|k-1}({\bf x}) &= \int g({\bf x}|{\bf x}^\prime) f_{k-1|k-1}({\bf x}^\prime) \delta {\bf x}^\prime,\label{eq_predict}\\
  f_{k|k}({\bf x}) &= \frac{\ell({\bf z}_k|{\bf x})f_{k|k-1}({\bf x})}{\int \ell({\bf z}_k|{\bf x})f_{k|k-1}({\bf x}) \delta {\bf x}}\label{eq_update}
\end{align}
where $g(\cdot|{\bf x})$ is the multi-object transition density and $\ell({\bf z}_k|\cdot)$ is the measurement likelihood. For the multi-object dynamic model described in Section \ref{dynamic_model} with Poisson birth model 
\begin{equation}
  \label{eq_poisson_birth}
  \beta({\bf x}) = e^{-\left\langle \lambda_k^B,1 \right\rangle} \left[ \lambda_k^B(\cdot) \right]^{{\bf x}}
\end{equation}
where $\lambda_k^B(\cdot)$ is the Poisson birth intensity at time step $k$, and general multi-object measurement models with Poisson clutter, if the prior is PMBM, the predicted and posterior densities on the current set of object states are PMBM \cite{garcia2021poisson}. We note that the Poisson birth model is a special case of PMBM.

The PMB is a common and efficient approximation of a PMBM \cite{pmbmpoint,variational}, and its filtering density is defined as
\begin{subequations}
  \label{eq_pmb_whole}
  \begin{align}
    f_{k^\prime|k}({\bf x}) &= \sum_{{\bf x}_d\uplus{\bf x}_u = {\bf x}}f^{p}_{k^\prime|k}({\bf x}_u)f^{mb}_{k^\prime|k}({\bf x}_d),\label{eq_pmb}\\
    f^{p}_{k^\prime|k}({\bf x}) &= e^{-\left\langle \lambda^u,1 \right\rangle} \left[ \lambda^u_{k^\prime|k}(\cdot) \right]^{{\bf x}},\\
    f^{mb}_{k^\prime|k}({\bf x}) &= \sum_{\uplus_{l=1}^{n_{k^\prime|k}}{\bf x}_d^l = {\bf x}}\prod_{i=1}^{n_{k^\prime|k}}f^i_{k^\prime|k}({\bf x}_d^i),\label{eq_mb}\\
    f^i_{k^\prime|k}({\bf x}) &= \begin{cases}
      1 - r^i_{k^\prime|k}, & {\bf x} = \emptyset \\
      r^i_{k^\prime|k}p^i_{k^\prime|k}(x), & {\bf x} = \{x\} \\
      0, & \text{otherwise}
    \end{cases}
  \end{align}
\end{subequations}
with $k^\prime\in\{k,k+1\}$. The PMB is the union of two independent RFSs: a PPP with density $f^p_{k^\prime|k}(\cdot)$, parameterised by Poisson intensity $\lambda^u_{k^\prime|k}(\cdot)$, representing undetected objects, and an MB process with density $f^{mb}_{k^\prime|k}(\cdot)$ representing detected objects where the $i$-th Bernoulli component has density $f^i_{k^\prime|k}(\cdot)$ with probability of existence $r^i_{k^\prime|k}$ and single-object state density $p^i_{k^\prime|k}(\cdot)$.


\section{Backward Smoothing for Sets of Trajectories}

In this paper, the objective is to compute the multi-trajectory posterior density $\pi_{1:K|K}({\bf X}_{1:K})$ using a sequence of multi-object filtering densities $f_{k|k}(\cdot)$ with $k = 1,\dots,K$ and the multi-trajectory dynamic model via backward smoothing. To achieve this, we first present the multistep prediction theorem for sets of trajectories, which generalises the general prediction theorem for sets of trajectories \cite[Theorem 7]{garcia2019multiple} to multistep prediction, along with a resulting corollary that is important for the derivation of the general backward smoothing equation for sets of trajectories. 

\begin{theorem}
  \label{thm_multipredicts}
  Given a multi-trajectory density $\pi_{\alpha:\eta|k}({\bf X}_{\alpha:\eta})$, its $(\gamma-\eta)$-step predicted multi-trajectory density, with $\alpha \leq \eta < \gamma$, $\eta \geq k$ and $\gamma \geq k+1$, is given by
  \begin{multline}
    \label{eq_multistep1}
    \pi_{\alpha:\gamma|k}({\bf X}_{\alpha:\gamma}) =  \prod_{\left(t,x^{1:\nu}\right)\in{\bf X}_{\alpha:\gamma}^{\eta}} \Bigg[ \left( 1+p^S\left(x^\nu\right) \left(\delta_{\gamma-t+1}[\nu]-1\right) \right)\\
    \times  \prod_{\ell=\eta-t+1}^{\nu-1} g \left( x^{\ell+1} | x^\ell\right) p^S\left(x^\ell\right) \Bigg] \pi_{\alpha:\eta|k}({\bf X}_{\alpha:\eta})\pi_{\eta+1:\gamma}({\bf W})
  \end{multline}
  \begin{multline}
    \label{eq_multistep2}
    \pi_{\eta+1:\gamma}({\bf W}) = \prod_{\left(t,x^{1:\nu}\right)\in{\bf W}}\Bigg[\left( 1+p^S\left(x^\nu\right)\left(\delta_{\gamma-t+1}[\nu]-1\right) \right)\\
    \times \prod_{\ell=1}^{\nu-1} g \left( x^{\ell+1} | x^\ell\right) p^S\left(x^\ell\right)  \Bigg] \prod_{\ell=\eta+1}^{\gamma} \beta\left(\tau^\ell\left({\bf W}^\ell\right)\right)
  \end{multline}
  where we write ${\bf X}_{\alpha:\gamma} = {\bf X}_{\alpha:\gamma}^{\eta} \uplus {\bf W}$ and ${\bf W} = {\bf W}^{\eta+1}\uplus \cdots \uplus {\bf W}^{\gamma}$ denotes the set of trajectories born in the time interval $\eta+1:\gamma$, where ${\bf W}^\ell = \left\{ \left(t,x^{1:\nu} \right)\in {\bf W}: t = \ell\right\}$ is the set of trajectories born at time step $\ell$ with $\eta + 1 \leq \ell \leq \gamma$.
\end{theorem}
Theorem \ref{thm_multipredicts} is proved in Appendix \ref{proof_thm_multipredicts}, and it describes that the $(\gamma-\eta)$-step predicted multi-trajectory density of $\pi_{\alpha:\eta|k}({\bf X}_{\alpha:\eta})$ can be evaluated by multiplying the following terms: the multi-trajectory density $\pi_{\alpha:\eta|k}({\bf X}_{\alpha:\eta})$, $1-p^S(\cdot)$ for trajectories that died in the time interval $\eta+1:\gamma$, $g(\cdot|\cdot)p^S(\cdot)$ for trajectories that were alive at time step $\eta$, and the multi-trajectory density $\pi_{\eta+1:\gamma}({\bf W})$ for trajectories that appeared after time step $\eta$.

\begin{corollary}
  \label{corollary_}
  For $\gamma \geq k+1$, it holds that
  \begin{equation}
    \label{eq_corollary}
    \frac{\pi_{k:\gamma|k}({\bf X}_{k:\gamma})}{\pi_{k+1:\gamma|k}({\bf X}_{k+1:\gamma})} = \frac{\pi_{k:k+1|k}({\bf X}_{k:k+1})}{f_{k+1|k}(\tau^{k+1}({\bf X}_{k+1:k+1}))}.
  \end{equation}
\end{corollary}
Corollary \ref{corollary_} is proved in Appendix \ref{proof_corollary_}, and it shows that the ratio between the two multi-trajectory densities $\pi_{k:\gamma|k}({\bf X}_{k:\gamma})$ and $\pi_{k+1:\gamma|k}({\bf X}_{k+1:\gamma})$ does not depend on $\gamma$ for $\gamma \geq k+1$.

The general backward smoothing equation for sets of trajectories under the conventional multi-object dynamic model assumptions is presented in the following theorem. Its proof is given in Appendix \ref{proof_thm_smoothing}.

\begin{theorem}
  \label{thm_smoothing}
  Given the multi-trajectory density $\pi_{k+1:K|K}(\cdot)$ and the multi-object state filtering density $f_{k|k}(\cdot)$, the multi-trajectory density in the time interval $k:K$ conditioned on the sequence of sets of measurements up to and including time step $K$ is
  \begin{equation}
    \label{eq_smoothing}
    \pi_{k:K|K}({\bf X}_{k:K}) = \frac{\pi_{k:k+1|k}({\bf X}_{k:k+1})\pi_{k+1:K|K}({\bf X}_{k+1:K})}{f_{k+1|k}(\tau^{k+1}({\bf X}_{k+1:k+1}))}
  \end{equation}
  where the multi-object state predicted density $f_{k+1|k}(\cdot)$ can be computed using $f_{k|k}(\cdot)$ via \eqref{eq_predict}, and $\pi_{k:k+1|k}({\bf X}_{k:k+1})$ is the one-step predicted multi-trajectory density of $\pi_{k:k|k}({\bf X}_{k:k})$, which takes the same value as $f_{k|k}(\tau^k({\bf X}_{k:k}))$.
\end{theorem}
Theorem \ref{thm_smoothing} shows that the multi-trajectory smoothing density $\pi_{k:K|K}({\bf X}_{k:K})$ can be expressed as the product of the multi-trajectory smoothing density $\pi_{k+1:K|K}({\bf X}_{k+1:K})$ and the one-step predicted multi-trajectory density of $f_{k|k}(\cdot)$, normalised by the multi-object state predicted density $f_{k+1|k}(\cdot)$. By applying Theorem \ref{thm_smoothing} recursively backwards in time, we obtain $\pi_{1:K|K}({\bf X}_{1:K})$, as desired. 


\section{A Multi-Trajectory Particle Smoother}


In this section, we first introduce a general backward kernel for sets of trajectories, which enables us to perform backward simulation and to approximate the multi-trajectory posterior. Then we present how to evaluate the backward kernel when the multi-object filtering densities are PMB. For the proposed multi-trajectory particle smoother using PMB filtering densities, the forward filtering representations involving single time step, single object marginal distributions can implicitly represent uncertainty over an exponentially large hypothesis space that cannot be feasibly represented through multiple hypothesis methods. The backward simulation framework allows us to draw samples in the multi-object trajectory space from these single time step, single object marginal distributions calculated during the forward filtering which incorporate information from measurements over all the time steps.

\subsection{Backward kernel for sets of trajectories}

The backward kernel for sets of trajectories is introduced in the following lemma. Its proof is given in Appendix \ref{proof_lemma_bs}.

\begin{lemma}
  \label{lemma_bs}
  The backward kernel for sets of trajectories, i.e., the multi-trajectory density of the set ${\bf X}$ of trajectories in the time interval $k:K$ conditioned on the set ${\bf Y}$ of trajectories in the time interval $k+1:K$ and the sequence of measurement sets up to and including time step $K$, satisfies
  \begin{equation}
    \label{eq_lemma_bs}
    \pi_{k:K|K}({\bf X}|{\bf Y}) \propto \pi_{k:k+1|k}({\bf X}_{k:k+1})\delta_{{\bf Y}}({\bf X}_{k+1:K}).
  \end{equation}
\end{lemma}

Lemma \ref{lemma_bs} describes that the backward kernel $\pi_{k:K|K}({\bf X}|{\bf Y})$ conditioned on the set ${\bf Y}$ of trajectories is proportional to the one-step predicted multi-trajectory density $\pi_{k:k+1|k}({\bf X}_{k:k+1})$ only if ${\bf Y} = {\bf X}_{k+1:K}$, and zero otherwise. Note that although the set ${\bf X}_{k+1:K}$ of trajectories is deterministically given by ${\bf Y}$, evaluating the backward kernel $\pi_{k:K|K}({\bf X}|{\bf Y})$ is complicated by the multiple different ways of associating ${\bf X}_{k:k+1}$ to ${\bf Y}$; see, e.g., Fig. \ref{fig_intro_example}.

In both Theorem \ref{thm_smoothing} and Lemma \ref{lemma_bs}, we consider the distribution of the set of trajectories in the time interval $k:K$, and we condition on all measurements in the time interval $1:K$. The main difference between Theorem \ref{thm_smoothing} and Lemma \ref{lemma_bs} is that the multi-trajectory density computed in Lemma \ref{lemma_bs} is conditioned on a particular value of the set ${\bf Y}$ of trajectories in the time interval $k + 1 : K$, which is what we require for backward simulation.

\subsection{Backward kernel for PMB filtering densities}
\label{sec_pmbm_recursion}
We proceed to present the backward kernel \eqref{eq_lemma_bs} for PMB filtering densities, see \eqref{eq_pmb_whole}. We first observe that the backward kernel \eqref{eq_lemma_bs} is in analogy to the Bayesian measurement update \eqref{eq_update} in the sense that $\pi_{k:k+1|k}({\bf X}_{k:k+1})$ is the prior and $\delta_{{\bf Y}}({\bf X}_{k+1:K})$ is the measurement likelihood. Specifically, for a PMB filtering density $f_{k|k}(\cdot)$, the one-step predicted multi-trajectory density $\pi_{k:k+1|k}({\bf X}_{k:k+1})$ is a trajectory PMB \cite[Lemma 4]{garcia2020trajectory}. Moreover, the multi-trajectory Dirac delta $\delta_{{\bf Y}}({\bf X}_{k+1:K})$ can be understood as a standard multi-object measurement model \cite{mahler2007statistical} with the following characteristics (this will be further elaborated in Appendix \ref{proof_thm_pmb}):
\begin{itemize}
  \item Each trajectory $X = (t,x^{1:\nu}) \in {\bf X}_{k:K}$ is detected with probability 
  \begin{equation}
    p^D(X) = \begin{cases}
      0, & t = k~\text{and}~\nu = 1\\
      1, & \text{otherwise}
    \end{cases}
  \end{equation}
  and if detected, it generates a measurement $Y$ with density $\delta_{Y}(X_{k+1:K})$.
  \item Clutter is Poisson with intensity $\lambda^C(\cdot) = 0$.
\end{itemize}
Then, the backward kernel $\pi_{k:K|K}({\bf X}|{\bf Y})$ \eqref{eq_lemma_bs} for PMB filtering densities, obtained via updating the trajectory PMB prior $\pi_{k:k+1|k}({\bf X}_{k:k+1})$ using $\delta_{{\bf Y}}({\bf X}_{k+1:K})$ is a trajectory PMBM \cite[Lemma 5]{garcia2020trajectory}. In what follows, we explain the result of the trajectory PMB update applied to backward simulation with PMB filtering densities, where each trajectory Bernoulli component has information on start/end times of the trajectories giving rise to the following local hypotheses:
\begin{itemize}
  \item The trajectory had never existed in the entire time interval $1:k+1$.
  \item The trajectory ended at time step $k$.
  \item The trajectory existed at both time step $k$ and $k+1$.
  \item The trajectory started at time $k+1$.
\end{itemize} 

We define by ${\cal M}_{k+1:K}= \{1,\dots,n_{k+1:K}\}$ the set of indices of trajectories in ${\bf Y}$. Each trajectory in ${\bf Y}$ creates a unique trajectory Bernoulli component. The number of trajectory Bernoulli components in $\pi_{k:K|K}({\bf X}|{\bf Y})$ is $n_{k:K|K} = n_{k|k} + n_{k+1:K}$, and they are indexed by variable $i\in\{1,\dots,n_{k:K|K}\}$. A global hypothesis is $a = (a^1,\dots,a^{n_{k:K|K}})$, where $h^i$ is the number of local hypotheses and $a^i\in\{1,\dots,h^i_{k:K|K}\}$ is the index to the local hypothesis for the $i$-th trajectory Bernoulli component, which incorporates the following information:
\begin{itemize}
  \item The set of indices ${\cal M}_{k+1:K}^{i,a^i} \subseteq {\cal M}_{k+1:K}$. If ${\cal M}_{k+1:K}^{i,a^i} = \emptyset$, then the $i$-th hypothesised trajectory did not appear after time step $k$. If ${\cal M}_{k+1:K}^{i,a^i} = \{j\}$, then the $i$-th hypothesised trajectory existed in the time interval $k+1:K$.
  \item The local hypothesis weight $w^{i,a^i}_{k:K|K}$ (used for computing the probability of the global hypotheses).
  \item The hypothesis-conditioned Bernoulli density $\pi^{i,a^i}_{k:K|K}(\cdot)$ of the form \eqref{eq_Bernoulli}, parameterised by the probability of existence $r^{i,a^i}_{k:K|K}$ and the single-trajectory density $p^{i,a^i}_{k:K|K}(\cdot)$.
\end{itemize}
The set of all global hypotheses is
\begin{multline}
  \label{eq_global_hypo}
  {\cal A}_{k:K|K} = \Bigg\{ \left(a^1,\dots,a^{n_{k:K|K}} \right) : a^i \in \left\{1,\dots,h^i_{k:K|K}\right\}~\forall~i,\\ \left|{\cal M}^{i,a^i}_k\right| \leq 1, \biguplus_{i=1}^{n_{k:K|K}}{\cal M}^{i,a^i}_k = {\cal M}_k\Bigg\},
\end{multline}
and the weight of global hypothesis $a\in{\cal A}_{k:K|K}$ satisfies 
\begin{equation}
  \label{eq_glo_hyo_weight}
  w^a \propto \prod_{i=1}^{n_{k:K|K}}w^{i,a^i}_{k:K|K}
\end{equation}
where the proportionality means that normalisation is required to ensure that $\sum_{a\in{\cal A}_{k:K|K}}w^a = 1$. A simple example of the hypothesis structure is illustrated and described in Fig. \ref{fig_illus}.

The following theorem gives the explicit expression of the PMBM backward kernel \eqref{eq_lemma_bs} for PMB filtering densities. 
\begin{figure}[!t]
  \centering
  \includegraphics[width=\columnwidth]{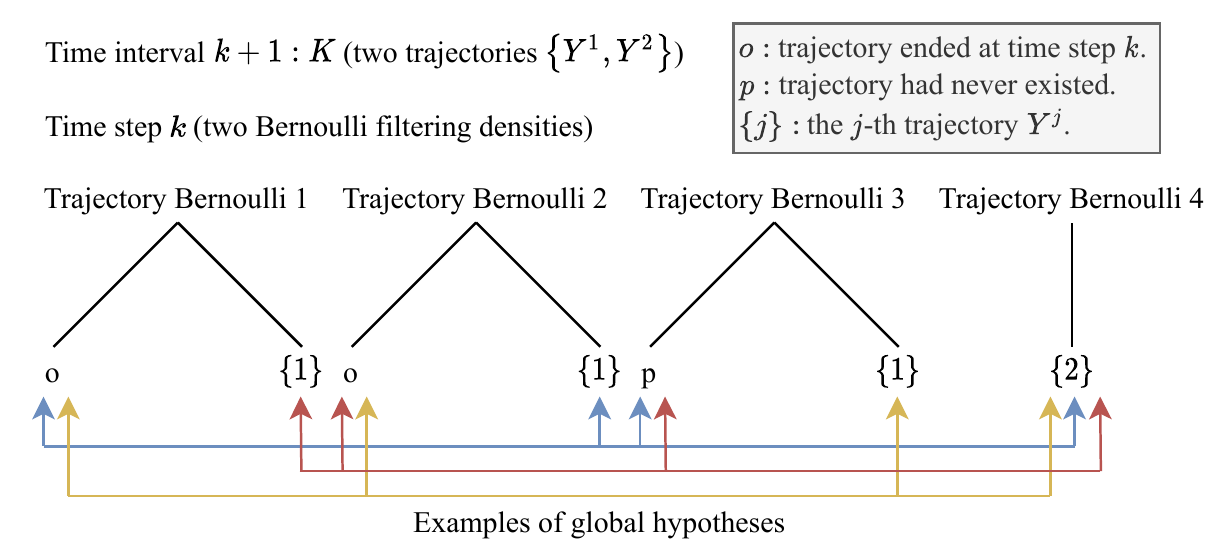}
  \caption{Hypothesis structure of the PMBM backward kernel for the example illustrated above where there are two trajectories in the time interval $k+1:K$ and two Bernoulli filtering densities at time step $k$. For the two trajectories: $Y^1$ is the trajectory of an object that existed at time step $k+1$, and $Y^2$ is the trajectory of an object that appeared after time step $k+1$. There are four trajectory Bernoulli components. For each trajectory Bernoulli component in $\pi_{k:k+1|k}({\bf X}_{k:k+1})$ (Trajectory Bernoulli 1 and 2), each has two local hypotheses: one corresponds to the case that the trajectory ended at time step $k$, and the other corresponds to the case that the trajectory Bernoulli component was updated by trajectory $Y^1$. The trajectory Bernoulli component created by trajectory $Y^1$ (Trajectory Bernoulli 3) has two local hypotheses: one corresponds to the case that the trajectory started at time step $k+1$, and the other corresponds to the case that the trajectory never existed.  The trajectory Bernoulli component created by trajectory $Y^2$ (Trajectory 4) has only a single local hypothesis since $Y^2$ remains unaltered.}
  \label{fig_illus}
\end{figure}

\begin{theorem}
  \label{thm_pmb}
  Given a PMB filtering density $f_{k|k}(\cdot)$ \eqref{eq_pmb_whole} at time step $k$, the set ${\bf Y}$ of trajectories in the time interval $k+1:K$ and the multi-trajectory dynamic model described in Section \ref{dynamic_model} with Poisson birth \eqref{eq_poisson_birth}, the multi-trajectory density in the time interval $k:K$ conditioned on ${\bf Y}$ and the sequence of measurement sets up to and including time step $K$, is a PMBM of the form
  \begin{align}
    \pi_{k:K|K}({\bf X}|{\bf Y}) &= \sum_{{\bf X}_u \uplus {\bf X}_d = {\bf X}} \pi_{k:K|K}^u({\bf X}_u)\pi_{k:K|K}^d({\bf X}_d|{\bf Y}),\label{eq_pmbm}\\
    \pi_{k:K|K}^u({\bf X}_u) &= e^{-\left\langle \lambda^u_{k:K|K},1 \right\rangle} \left[ \lambda^u_{k:K|K}(\cdot) \right]^{{\bf X}_u},
  \end{align}
  \begin{multline}
    \pi_{k:K|K}^d({\bf X}_d|{\bf Y}) = \\ \sum_{a\in {\cal A}_{k:K|K}}w^a \sum_{\biguplus^{n_{k:K|K}}_{j=1} {\bf X}^j = {\bf X}_d} \prod_{i=1}^{n_{k:K|K}} \pi^{i,a^i}_{k:K|K}\left({\bf X}^i\right)\label{eq_mbm}
  \end{multline}
  where $n_{k:K|K} = n_{k+1:K} + n_{k|k}$, $w^a$ is given by \eqref{eq_glo_hyo_weight} and 
  \begin{equation}
    \lambda^u_{k:K|K}\left(t,x^{1:\nu}\right) = \delta_{k}[t]\delta_{1}[\nu]\left(1-p^S\left(x^1\right)\right)\lambda^u_{k|k}\left(x^1\right).
  \end{equation}
  
  Also, we write ${\bf Y} = \{Y^1,\dots,Y^{n_{k+1:K}}\}$ with $Y^1,\dots,Y^m$ being trajectories of objects that existed at time step $k+1$, and $Y^m,\dots,Y^{n_{k+1:K}}$ being trajectories of objects that appeared after time step $k+1$. 
  
  For each trajectory Bernoulli component in the predicted trajectory PMB $\pi_{k:k+1|k}({\bf X}_{k:k+1})$, $i\in\{1,\dots,n_{k|k}\}$, there are $h_{k:K|K}^i = m+1$ local hypotheses. The local hypothesis, corresponding to the case that the trajectory ended at time step $k$, is given by ${\cal M}_{k+1:K}^{i,1} = \emptyset$ and 
  \begin{subequations}
    \begin{align}
      w^{i,1}_{k:K|K} &= 1-r^{i}_{k|k}+r^i_{k|k}\left\langle p^i_{k|k},1-p^S \right\rangle,\\
      r^{i,1}_{k:K|K} &= \frac{r^i_{k|k}\left\langle p^i_{k|k},1-p^S \right\rangle}{w^{i,1}_{k:K|K}},\\
      p^{i,1}_{k:K|K}\left(t,x^{1:\nu}\right) &= \delta_k[t]\delta_{1}[\nu]\frac{p^i_{k|k}\left(x^1\right)\left(1-p^S\left(x^1\right)\right)}{\left\langle p^i_{k|k},1-p^S \right\rangle}.
    \end{align}
  \end{subequations}
  The local hypothesis, corresponding to the case that the trajectory Bernoulli component is updated by trajectory $Y^j = \left(t^j,y^{1:\nu^j}\right)$, $j\in\{1,\dots,m\}$ (present at time step $k+1$, i.e., $t^j=k+1$), is given by ${\cal M}_{k+1:K}^{i,j+1} = \{j\}$ and
  \begin{subequations}
    \begin{align}
      w^{i,j+1}_{k:K|K} &= r^i_{k|k}\left\langle p^i_{k|k},g\left(y^1|\cdot\right)p^S \right\rangle,\\
      r^{i,j+1}_{k:K|K} &= 1,\\
      p^{i,j+1}_{k:K|K}\left(t,x^{1:\nu}\right) &=\delta_{k}[t]\delta_{\nu^j+1}[\nu]\delta_{y^{1:\nu^j}}\left(x^{2:\nu}\right)\nonumber\\ &~~~\times\frac{g\left(y^1|x^1\right)p^i_{k|k}\left(x^1\right)p^S\left(x^1\right)}{\left\langle p^i_{k|k},g\left(y^1|\cdot\right)p^S \right\rangle}.
    \end{align}
  \end{subequations}

  The trajectory Bernoulli component created by trajectory $Y^j = \left(t^j,y^{1:\nu^j}\right)$, $j\in\{1,\dots,m\}$, 
  has two local hypotheses $h^i_{k:K|K} = 2$. The first one corresponds to a non-existent Bernoulli and is given by ${\cal M}_{k+1:K}^{i,1} = \emptyset$ and 
  \begin{equation}
    \label{eq_non_valid}
    w^{i,1}_{k:K|K} = 1,\quad r_{k:K|K}^{i,1} = 0.
  \end{equation}
  The second one is given by ${\cal M}_{k+1:K}^{i,2} = \{j\}$ and
  \begin{subequations}
    \label{eq_newly_detected}
    \begin{align}
      w^{i,2}_{k:K|K} &= \lambda^B_{k+1}\left(y^1\right) + \left\langle \lambda^u_{k|k},g\left(y^1|\cdot\right)p^S \right\rangle,\\
      r^{i,2}_{k:K|K} &= 1,\\
      p^{i,2}_{k:K|K}\left(t,x^{1:\nu}\right) &= \underline{w}^{i,2}_{k:K|K}\delta_{\left(t^j,y^{1:\nu^j}\right)}\left(t,x^{1:\nu}\right)\nonumber\\
      &~~~+\overline{w}^{i,2}_{k:K|K}\delta_{k}[t]\delta_{\nu^j+1}[\nu]\delta_{y^{1:\nu^j}}\left(x^{2:\nu}\right)\nonumber \\
      &~~~\times \frac{g\left(y^1|x^1\right)\lambda^u_{k|k}\left(x^1\right)p^S\left(x^1\right)}{\left\langle \lambda^u_{k|k},g\left(y^1|\cdot\right)p^S \right\rangle},\\
      \underline{w}^{i,2}_{k:K|K} &= \frac{\lambda^B_{k+1}\left(y^1\right)}{w^{i,2}_{k:K|K}},\\
      \overline{w}^{i,2}_{k:K|K} &= 1- \underline{w}^{i,2}_{k:K|K}.
    \end{align}
  \end{subequations}

  The trajectory Bernoulli component created by trajectory $Y^j$, $j\in\{m+1,\dots,n_{k=1:K}\}$ (not present at time step $k+1$) only has a single local hypothesis $h_{k:K|K}^i = 1$, given by ${\cal M}_{k+1:K}^{i,1} = \{j\}$ and
  \begin{subequations}
    \label{eq_pmb_theorem_d}
    \begin{align}
      w^{i,1}_{k:K|K} &= 1,\\
      r^{i,1}_{k:K|K} &= 1,\\
      p^{i,1}_{k:K|K}(X) &= \delta_{Y^j}(X).
    \end{align}
  \end{subequations}
\end{theorem}
\begin{proof}
  See Appendix \ref{proof_thm_pmb}.
\end{proof}

The local hypothesis parameterised by \eqref{eq_non_valid} corresponds to a non-existent Bernoulli component. For a global hypothesis that includes this local hypothesis, trajectory $Y^j$ has been assigned to a Bernoulli filtering density. The local hypothesis parameterised by \eqref{eq_newly_detected} covers the case that the trajectory is associated to the trajectory PPP in $\pi_{k:k+1|k}({\bf X}_{k:k+1})$; i.e., the corresponding object was first detected at time step $k+1$. Note that, for an object that was first detected at time step $k+1$, it might be born at any time before, or at, time step $k+1$. Therefore, there is uncertainty on the start time of its corresponding trajectory. 



\subsection{Backward simulation for sets of trajectories}

A particle approximation of the multi-trajectory density is
\begin{equation}
  \label{eq_particle}
  \pi({\bf X}_{\alpha:\gamma}) \approx \sum_{i=1}^{T}w^{(i)}\delta_{{\bf X}^{(i)}}({\bf X}_{\alpha:\gamma})
\end{equation}
where $T$ is the number of particles and the $i$-th particle has state ${\bf X}^{(i)}$ and weight $w^{(i)}$. By running the backward simulation for sets of trajectories $T$ times for $k = K-1,\dots,1$ where we recursively draw samples of ${\bf X}_{k:K}$ from the backward kernel \eqref{eq_lemma_bs}, we can obtain $T$ particles $\{{\bf X}^{(i)}_{1:K}\}_{i=1}^T$ representing the multi-trajectory density $\pi_{1:K|K}({\bf X}_{1:K})$ with uniform weights $w^{(i)} = 1/T$. Note that, multiple backward set of trajectories can be generated independently, without having to rerun the forward multi-object filter. That is the complexity of backward simulation for sets of trajectories is linear with the number of particles.


For the PMBM backward kernel \eqref{eq_pmbm}, it is sufficient to sample sets of trajectories described by the MBM \eqref{eq_mbm}, which are trajectories of objects that are hypothesised to have been detected at some time, possibly also including object states before their first detection. In other words, we do not sample trajectories of objects that are hypothesised to exist but were never detected. To do so, we first sample a global hypothesis to obtain a trajectory MB, and then we sample a single trajectory state for each Bernoulli component. It is, however, generally intractable to exhaustively enumerate all the global hypotheses. Therefore, approximations are needed for a tractable implementation. One such implementation will be presented for a linear-Gaussian dynamic model in the next section.


\section{A Tractable Implementation for Linear-Gaussian Dynamic Model}

In this section, we present a tractable implementation of the proposed multi-trajectory particle smoother for PMB filtering densities with the following assumptions:
\begin{assumption}
\label{assumption}
The linear Gaussian dynamic model and the PMB filtering densities are defined as follows
\begin{itemize}
  \item The survival probabilities are constant, i.e., $p^S(\cdot) = p^S$.
  \item The linear Gaussian single-object state transition density is $g(\cdot|x) = {\cal N}(\cdot;Fx,Q)$ where $F$ is the state transition matrix and $Q$ is the motion noise covariance matrix. 
  \item The Poisson birth intensity is a Gaussian mixture
  \begin{equation}
    \lambda^B_k(x) = \sum_{i=1}^{N^b_k}w_k^{b,i}{\cal N}\left(x;x_k^{b,i},P_k^{b,i}\right)\label{eq_Gaussian_birth}
  \end{equation}
  where $N^b_k$ is the number of components, $w_k^{b,i}$, $x_k^{b,i}$ and $P_k^{b,i}$ are the weight, the mean and the covariance of the $i$-th component, respectively.
  \item The PMB filtering density at time step $k$ is parameterised by $\left\{\lambda^u_{k|k}(\cdot), \left\{r^i_{k|k},p^i_{k|k}(\cdot)\right\}_{i=1}^{n_{k|k}}\right\}$. Here, 
  \begin{equation}
    \lambda^u_{k|k}(x) = \sum_{i=1}^{N_{k|k}^{u}}w_{k|k}^{u,i}{\cal N}\left(x;x_{k|k}^{u,i},P_{k|k}^{u,i}\right)\label{eq_Gaussian_ppp}
  \end{equation}
  is the Poisson RFS intensity for undetected objects where $N^u_{k|k}$ is the number of components, $w_{k|k}^{u,i}$, $x_{k|k}^{u,i}$ and $P_{k|k}^{u,i}$ are the weight, the mean and the covariance of the $i$-th component, respectively. In addition, there are $n_{k|k}$ Bernoulli components, and the $i$-th component has probability of existence $r^i_{k|k}$ and single-object state density $p^i_{k|k}(x) = {\cal N}\left(x;x^{i}_{k|k},P^{i}_{k|k}\right)$
  where $x_{k|k}^{i}$ is its mean and $P_{k|k}^{i}$ is its covariance.
\end{itemize}
\end{assumption}

\subsection{Gaussian implementation for backward kernel}

The backward kernel for sets of trajectories under Assumption \ref{assumption} is given by the following lemma. 
For general non-linear dynamic models, the smoothed state density can be computed using Gaussian assumed density approximations \cite{sarkka2013bayesian}.
\begin{lemma}
  \label{lemma_pmbm}
  Given the linear Gaussian dynamic model, the PMB filtering density at time step $k$ specified in Assumption \ref{assumption} and the set ${\bf Y} = \{Y^1,\dots,Y^{n_{k+1:K}}\}$ of trajectories in the time interval $k+1:K$ specified in Theorem \ref{thm_pmb}, the multi-trajectory density in the time interval $k:K$ is a PMBM of the form \eqref{eq_pmbm} with $n_{k:K|K} = n_{k+1:K} + n_{k|k}$ and Poisson intensity
  \begin{equation}
    \lambda^u_{k:K|K}\left(t,x^{1:\nu}\right) = \delta_{k}[t]\delta_{1}[\nu]\left(1-p^S\right)\lambda^u_{k|k}\left(x^1\right).
  \end{equation}
  The local hypothesis of the $i$-th trajectory Bernoulli component, $i\in\{1,\dots,n_{k|k}\}$, corresponding to the case that the trajectory ended at time step $k$, is given by 
  \begin{subequations}
    \begin{align}
      w^{i,1}_{k:K|K} &= 1 - r^i_{k|k} + r^i_{k|k}\left(1-p^S\right),\label{eq_W2}\\
      r^{i,1}_{k:K|K} &= \frac{r^i_{k|k}\left(1-p^S\right)}{w^{i,1}_{k:K|K}},\label{eq_mis_r}\\
      p^{i,1}_{k:K|K}\left(t,x^{1:\nu}\right) &= \delta_k[t]\delta_{1}[\nu]{\cal N}\left(x^1;x^{i}_{k|k},P^{i}_{k|k}\right)\label{eq_mis}.
    \end{align}
  \end{subequations}
  The local hypothesis of the $i$-th trajectory Bernoulli component, $i\in\{1,\dots,n_{k|k}\}$, corresponding to the case that the trajectory Bernoulli component is updated by trajectory $Y^j = \left(t^j,y^{1:\nu^j}\right)$, $j\in\{1,\dots,m\}$, is given by 
  \begin{subequations}
    \begin{align}
      w^{i,j+1}_{k:K|K} &= r^i_{k|k}p^S{\cal N}\left(y^1;Fx_{k|k}^i,P_{k+1|k}^i\right),\label{eq_W1_2}\\
      r^{i,j+1}_{k:K|K} &= 1,\\
      p^{i,j+1}_{k:K|K}\left(t,x^{1:\nu}\right) &=\delta_{k}[t]\delta_{\nu^j+1}[\nu]\delta_{y^{1:\nu^j}}\left(x^{2:\nu}\right)\nonumber\\ &~~~\times{\cal N}\left(x^1;x_{k|K}^{i},P_{k|K}^{i}\right),\label{eq_Gaussian_cov2}\\
      x^{i}_{k|K} &= x^{i}_{k|k} + G^i\left(y^1 - Fx^{i}_{k|k}\right),\\
      P^{i}_{k|K} &= P^{i}_{k|k} - G^iFP^{i}_{k|k},\label{eq_Gaussian_P}\\
      G^i &= P^{i}_{k|k}F^T\left(P_{k+1|k}^{i}\right)^{-1},\\
      P_{k+1|k}^{i} &= FP_{k|k}^{i}F^T+Q.
    \end{align}
  \end{subequations}
  The local hypothesis of the $i$-th trajectory Bernoulli component, $i\in\{n_{k|k}+1,\dots,n_{k|k}+m\}$, corresponding to the case that the object with trajectory $Y^j = \left(t^j,y^{1:\nu^j}\right)$ was first detected at time step $k+1$, is given by
  \begin{subequations}
    \label{eq_newly_detected_Gaussian}
    \begin{align}
      w^{i,2}_{k:K|K} &=\sum_{i=1}^{N^b_{k+1}}w_{k+1}^{b,i}{\cal N}\left(y^1;x_{k+1}^{b,i},P_{k+1}^{b,i}\right)\nonumber\\ &~~~+ p^S\sum_{i=1}^{N_{k|k}^{u}}w_{k|k}^{u,i}{\cal N}\left(y^1;Fx_{k|k}^{u,i},P_{k+1|k}^{u,i}\right),\label{eq_W1_1}\\
      r_{k:K|K}^{i,2} &= 1,\\
      p^{i,2}_{k:K|K}\left(t,x^{1:\nu}\right) &= \underline{w}^{i,2}_{k:K|K}\delta_{\left(t^j,y^{1:\nu^j}\right)}\left(t,x^{1:\nu}\right)\nonumber\\
      &~~~+\overline{w}^{i,2}_{k:K|K}\delta_{k}[t]\delta_{\nu^j+1}[\nu]\delta_{y^{1:\nu^j}}\left(x^{2:\nu}\right)\nonumber \\
      &~~~\times \frac{\sum_{i=1}^{N^u_{k|k}}w_{k|k}^{u,i}{\cal N}\left(x^1;x^{u,i}_{k|K},P^{u,i}_{k|K}\right)}{\sum_{i=1}^{N^u_{k|k}}w_{k|k}^{u,i}},\label{eq_lemma_ppp}\\
      x^{u,i}_{k|K} &= x^{u,i}_{k|k} + G^i\left(y^1 - Fx^{u,i}_{k|k}\right),\\
      P^{u,i}_{k|K} &= P^{u,i}_{k|k} - G^iFP^{u,i}_{k|k},\\
      G^i &= P^{u,i}_{k|k}F^T\left(P_{k+1|k}^{u,i}\right)^{-1},\\
      P_{k+1|k}^{u,i} &= FP_{k|k}^{u,i}F^T+Q,\\
      \underline{w}^{i,2}_{k:K|K} &= \frac{\sum_{i=1}^{N^b_{k+1}}w_{k+1}^{b,i}{\cal N}\left(y^1;x_{k+1}^{b,i},P_{k+1}^{b,i}\right)}{w^{i,2}_{k:K|K}},\label{eq_ww1}\\ 
      \overline{w}^{i,2}_{k:K|K} &= 1- \underline{w}^{i,2}_{k:K|K}\label{eq_ww2}.
    \end{align}
  \end{subequations}
  The local hypothesis of the $i$-th trajectory Bernoulli component, $i\in\{n_{k|k}+j\}$, $j\in\{m+1,\dots,n_{k+1:K}\}$, corresponding to the case that trajectory $Y^j$ remains unaltered, has $w^{i,1}_{k:K|K} = 1$, $r^{i,1}_{k:K|K} = 1$ and $p^{i,1}_{k:K|K}(X) = \delta_{Y^j}(X)$. 
\end{lemma}


\subsection{Practical considerations}

The difficulty in drawing samples of sets of trajectories from the trajectory PMBM backward kernel \eqref{eq_pmbm} is the number of global hypotheses since enumerating every global hypothesis is of combinatorial complexity. A simple solution is to truncate the MB mixture \eqref{eq_mbm} to only keep global hypotheses with non-negligible weights. This can be achieved by solving a ranked assignment problem using, e.g., Murty's algorithm \cite{crouse2016implementing}. 

As trajectory Bernoulli components created by trajectories $\{Y^j\}_{j=m+1}^{n_{k+1:K}}$ not present at time step $k+1$ all have a single local hypothesis with weight $w^{i,1}_{k:K|K} = 1$, $i\in\{n_{k|k}+m+1,\dots,n_{k:K|K}\}$, the global hypothesis weight \eqref{eq_glo_hyo_weight} becomes
\begin{align}
  w^a &\propto \prod_{i=1}^{n_{k|k}} w^{i,a^i}_{k:K|K}\prod_{i={n_{k|k}}+1}^{n_{k|k}+m}w^{i,a^i}_{k:K|K}\nonumber\\
  &\propto \prod_{i=1:a^i>1}^{n_{k|k}} \frac{w^{i,a^i}_{k:K|K}}{w^{i,1}_{k:K|K}}\prod_{i={n_{k|k}}+1}^{n_{k|k}+m}w^{i,a^i}_{k:K|K}.\label{eq_glo_hyo_weight2}
\end{align}
We can then construct the corresponding cost matrix as:
\begin{subequations}
\begin{align}
  C &= -\log \begin{bmatrix}
    W_1 & W_2
  \end{bmatrix},\\
  W^{(j,i)}_1 &= \frac{w^{i,j+1}_{k:K|K}}{w^{i,1}_{k:K|K}},\label{eq_weight_matrix_entry}\\
  W_2 &= \text{diag}\left( w_{k:K|K}^{{n_{k|k}}+1,2},\dots,w_{k:K|K}^{n_{k|k}+m,2} \right)\label{eq_weight_matrix_entry2}
\end{align}
\end{subequations}
where $W_2\in \mathbb{R}^{m\times m}$ is a diagonal matrix, $W_1 \in \mathbb{R}^{m\times n_{k|k}}$ and its $(j,i)$-th entry is $W_1^{(j,i)}$, the weight of associating the $j$-th trajectory to the $i$-th trajectory Bernoulli component. Each global hypothesis can be represented as an $m \times (n_{k|k}+m)$ assignment matrix $S$ consisting of $0$ or $1$ entries such that each row sums to one and that each column sums to zero or one. Then we obtain the $M$-best global hypotheses that minimise $\text{tr}(S^TC)$ using Murty's algorithm. Note that the assignment problem formulated here is similar to the assignment problem involved in the PMBM filter update where measurements are associated to Bernoulli/PPP components.

In addition to pruning global hypotheses with small weights, we apply ellipsoidal gating to remove unlikely local hypotheses. Specifically, if the squared Mahalanobis distance between the first state $y^1$ of trajectory $Y^j$, $1 \leq j \leq m$ and the predicted density of ${\cal N}(x;x^i_{k|k},P^i_{k|k})$, i.e.,
\begin{multline*}
  \text{SMD}(y^1;x_{k|k}^i,P_{k|k}^i) \triangleq \\ 
  \left( y^1 - Fx^i_{k|k} \right)^T \left( FP^i_{k|k}F^T + Q \right)^{-1} \left( y^1 - Fx^i_{k|k} \right)
\end{multline*}
is larger than a pre-defined threshold $\Gamma_g$, we then set $W^{(j,i)}_1 = 0$. This also applies to the Gaussian components of \eqref{eq_Gaussian_birth} and \eqref{eq_Gaussian_ppp}. In addition, clustering can be used to further simplify the computations of the assignment problem.

Having discussed strategies on how to efficiently sample the global hypotheses, we proceed to describe how to simplify the sampling of trajectories from the trajectory Bernoulli densities. We note that a fairly large number of particles may be needed to find the mode of the multi-trajectory density $\pi_{1:K|K}({\bf X})$ when the Gaussian covariance matrices in \eqref{eq_lemma_ppp} and \eqref{eq_Gaussian_cov2} are large. A simple heuristic that helps to quickly find the mode in some situations is to approximate each of these Gaussian densities as a Dirac delta centred at its mean. For instance, in this case the single-trajectory density \eqref{eq_Gaussian_cov2} becomes
\begin{equation*}
  \hat{p}^{i,j+1}_{k:K|K}\left(t,x^{1:\nu}\right) =\delta_{k}[t]\delta_{\nu^j+1}[\nu]\delta_{x_{k|K}^i}\left(x^1\right)\delta_{y^{1:\nu^j}}\left(x^{2:\nu}\right),
\end{equation*}
and this also avoids the need for computing the corresponding Gaussian covariance $P_{k|K}^i$ \eqref{eq_Gaussian_P}.

The robustness of the proposed multi-trajectory particle smoother mainly depends on the number $T$ of particles and the maximum number $M$ of global hypotheses used in backward simulation. For scenarios with large motion noises or PMB filtering densities with many Bernoulli components, we may need large $T$ and $M$ to obtain good smoothing performance. In fact, if the computational complexity is not a concern in offline applications (which is usually the case), one can always use larger $T$ and $M$ to obtain improved results.

Finally, the pseudocode of linear/Gaussian backward simulation for sets of trajectories with PMB filtering densities can be found in Appendix \ref{pseudo_code}.

\section{Illustrative Example and Simulation Results}

In this section, we first illustrate how to perform backward smoothing for sets of trajectories using Theorem \ref{thm_smoothing} via a one-dimensional example shown in Fig. \ref{fig_1d_example_gt}. Then, we evaluate the performance of the proposed multi-trajectory particle smoother for PMB filtering densities in two different two-dimensional scenarios.

\subsection{Illustrative example}
\label{sec_example}
We consider a one-dimensional scenario on the line segment $[-1\text{~m},1 \text{~m}]$, of length $K=8$ time steps, where two initially well-separated objects first approach each other (on the real line), then pause for $1\text{~s}$, before crossing each other. In the model, objects survive with probability $p^S = 0.95$ and new objects appear according to a Poisson birth model with intensity $\lambda^B(x) = 0.1\times\text{Uniform}[-1,1]$. A nearly constant velocity motion model is used with sampling period $T_s = 1 \text{~s}$ and single-object state transition density parameterised by
\begin{equation*}
  F = \begin{bmatrix}
    1 & T_s\\
    0 & 1
  \end{bmatrix},\quad Q = \frac{1}{900}\begin{bmatrix}
    T_s^3/3 & T_s^2/2\\
    T_s^2/2 & T_s
  \end{bmatrix}
\end{equation*}
We also assume that the multi-object filtering density at each time step is a multi-object Dirac delta, which is a special type of PMB with zero PPP intensity and each Bernoulli component having probability of existence one and single target density being a Dirac delta. The ground truth and filtering estimates are illustrated in Fig. \ref{fig_1d_example_gt}.

\begin{figure*}[!t]
  \centering
  \subfloat[\label{fig_1d_example_gt}]
  {\includegraphics[width=0.5\columnwidth]{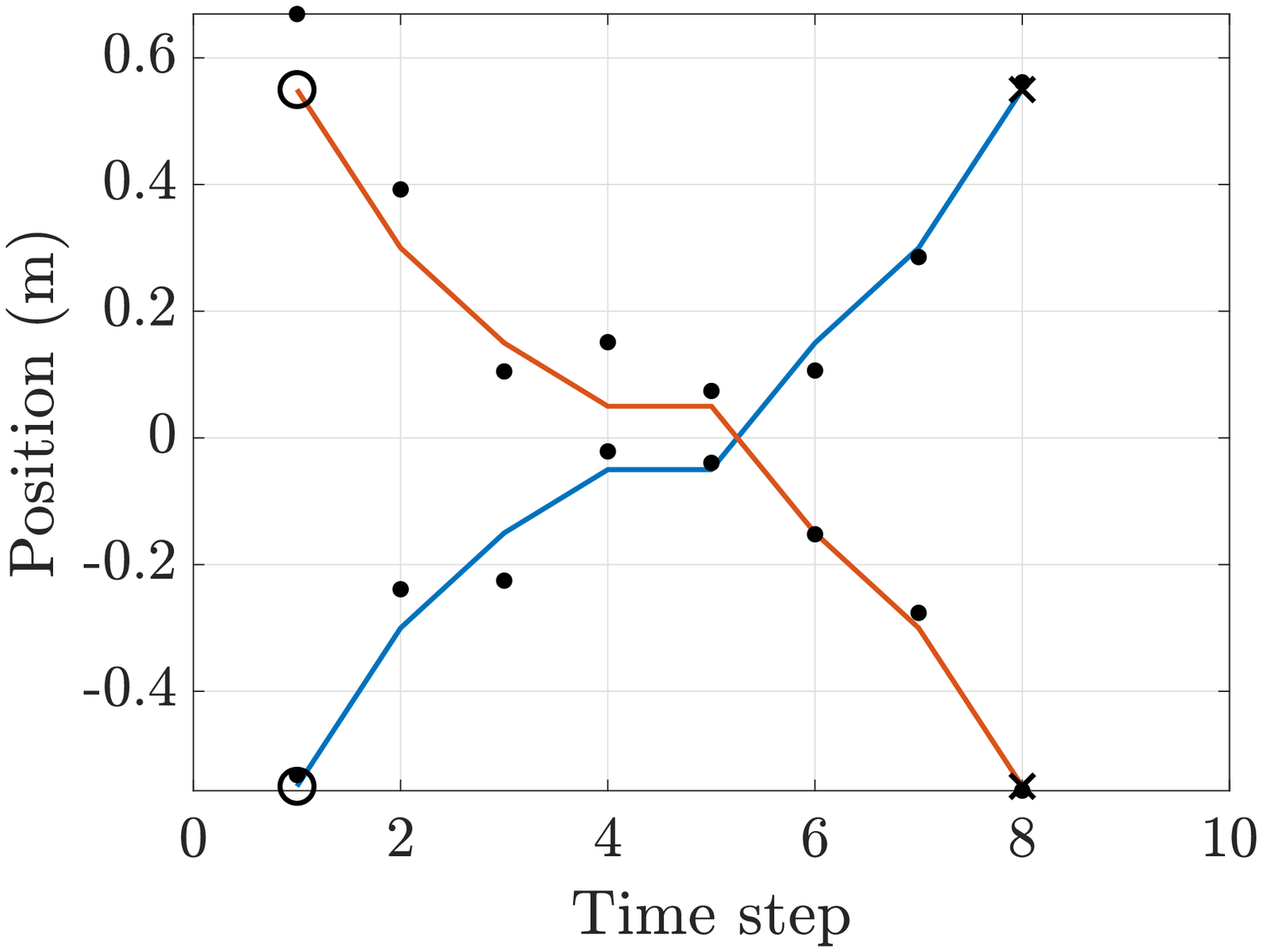}}
  \subfloat[\label{fig_1d_example_est1}]
  {\includegraphics[width=0.5\columnwidth]{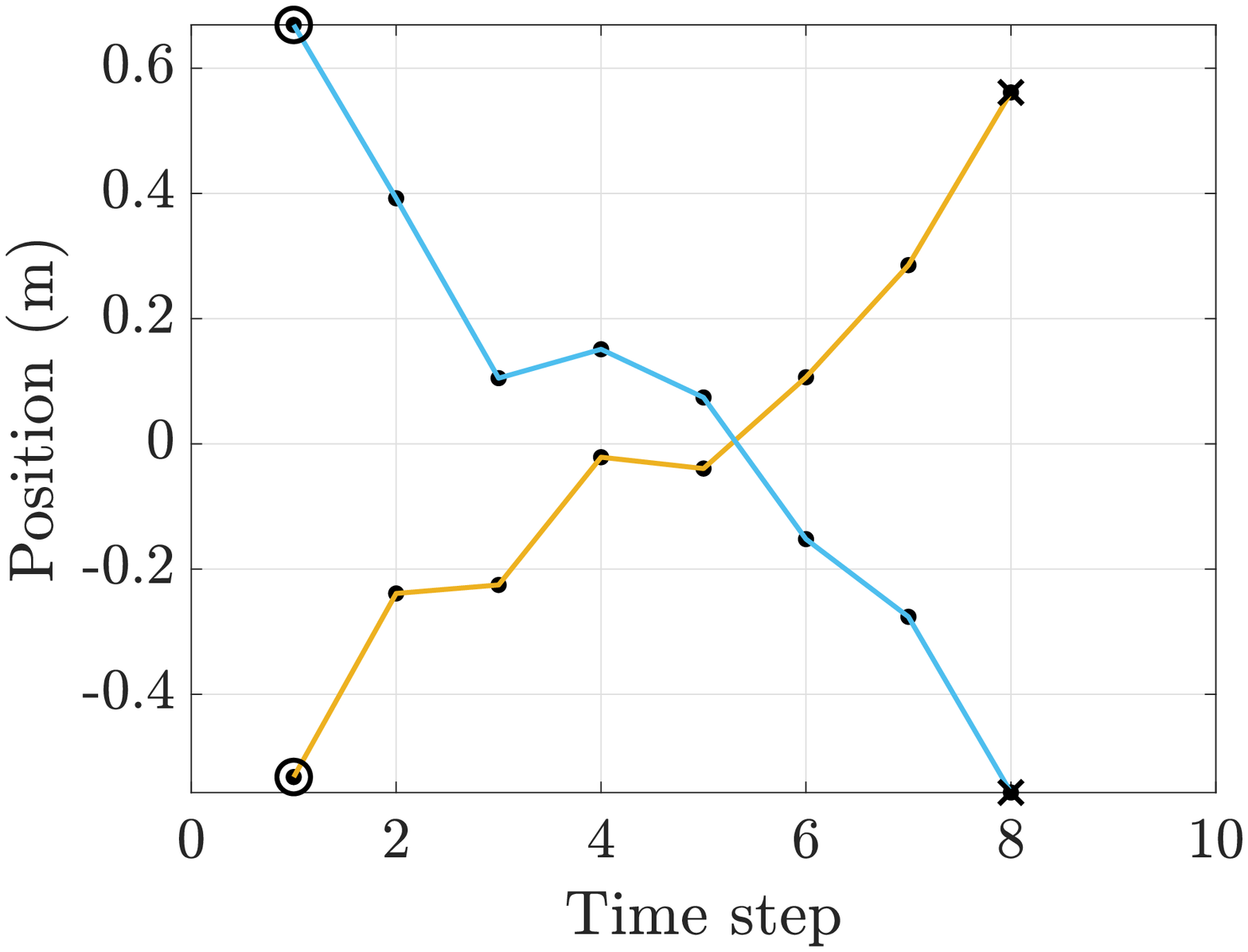}}
  \subfloat[\label{fig_1d_example_est2}]
  {\includegraphics[width=0.5\columnwidth]{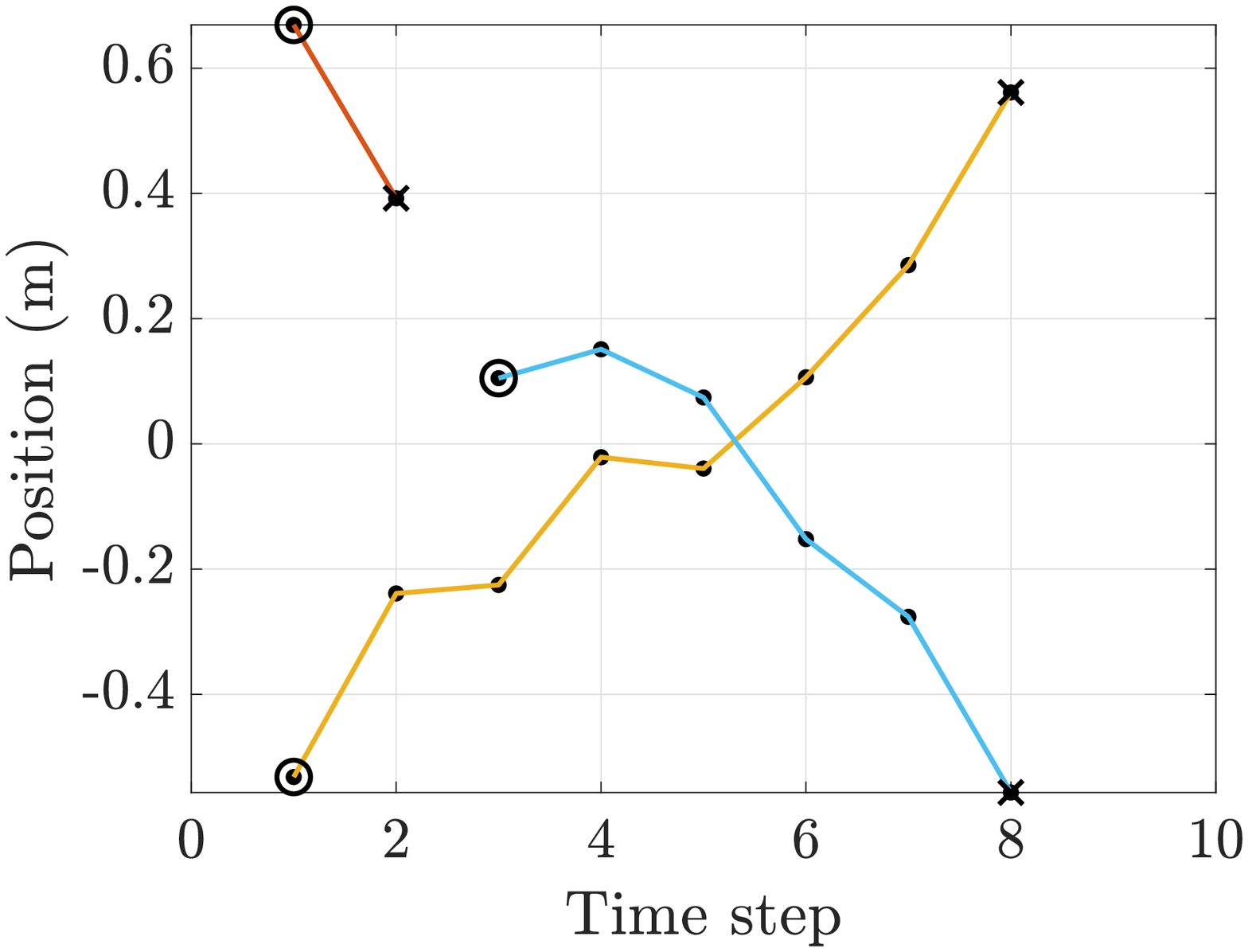}}
  \subfloat[\label{fig_1d_example_est3}]
  {\includegraphics[width=0.5\columnwidth]{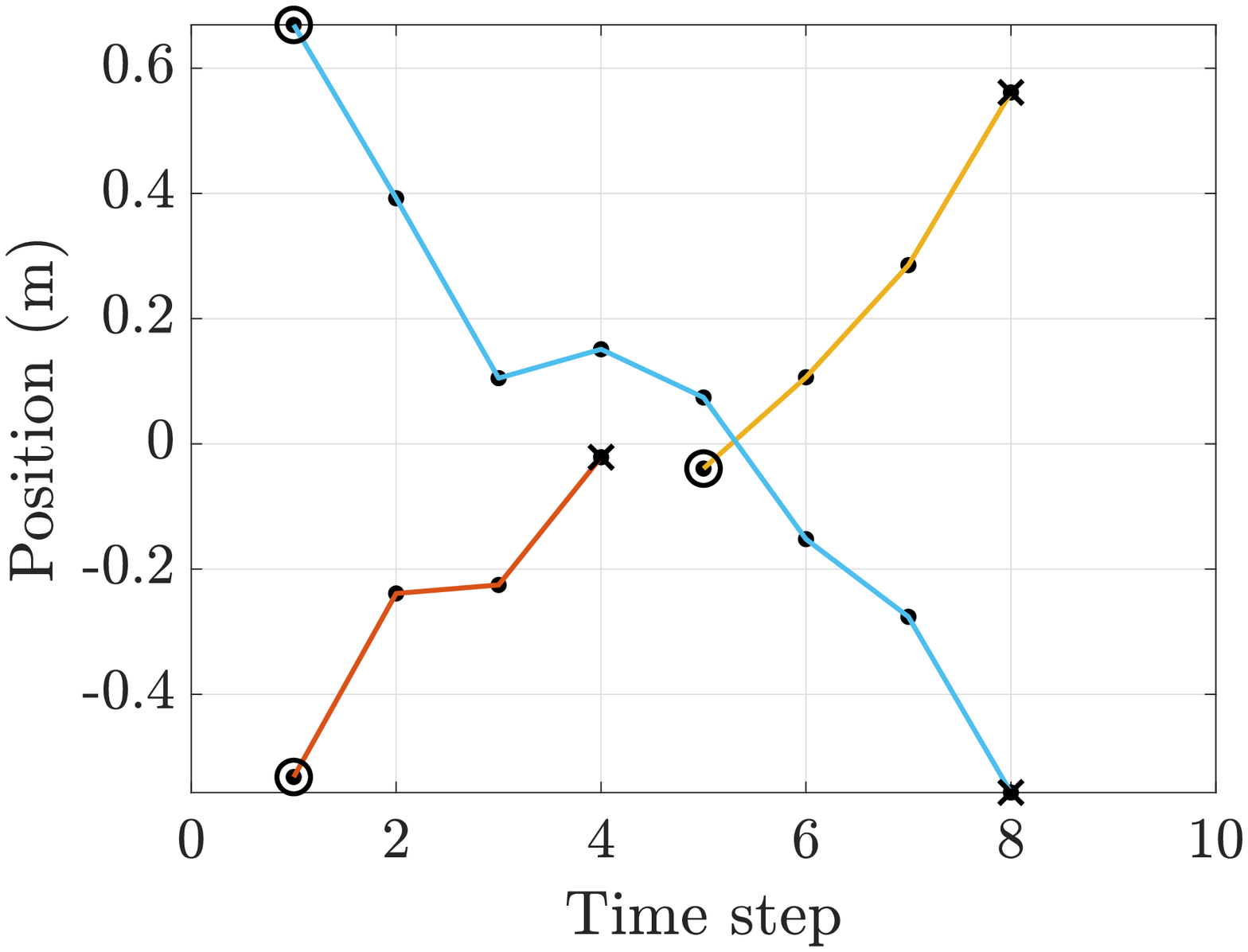}}
  \caption{Single-object Dirac deltas and true/estimated trajectories of a one-dimensional illustrative example. The single-object Dirac deltas are shown as black dots. Start/end positions of trajectories are marked by circles/crosses. Subfigure (a) shows the true trajectories. Subfigure (b), (c) and (d) show the estimated sets of trajectories in descending order of posterior probability; their posterior probabilities are $0.78$, $0.15$ and $0.03$, respectively.}
  \label{fig_1d_example}
\end{figure*}

The objective is to compute the multi-trajectory posterior $\pi_{1:K|K}({\bf X}_{1:K})$, which, in this example, can be directly derived using Theorem \ref{thm_smoothing}. However, due to the multiple possible associations of single-object Dirac deltas at different time steps, the number of mixture components in the multi-trajectory smoothing density $\pi_{k:K|K}({\bf X}_{k:K})$ increases very fast when computed backward in time, and thus it soon becomes infeasible to evaluate $\pi_{k:K|K}({\bf X}_{k:K})$ without approximation. As shown in Appendix \ref{appendix_example_expression}, the multi-trajectory smoothing density $\pi_{K-1:K|K}({\bf X}_{K-1:K})$ has already a mixture representation of 7 components. In this example, we compute the approximate $\pi_{1:K|K}({\bf X}_{1:K})$ by pruning mixture components with negligible weights. Fig. \ref{fig_1d_example_est1}, \ref{fig_1d_example_est2} and \ref{fig_1d_example_est3} show the estimated sets of trajectories in descending order of posterior probability. The results show that the estimate with unbroken trajectories has the highest weight. With the considered multi-object dynamic model, the set of trajectories with the highest posterior probability closely matches the true trajectories. The other sets of trajectories (with less probability) do not properly link one of the trajectories and, therefore, estimate an additional trajectory.

\subsection{Simulation results}

We present the results from a Monte Carlo simulation with 500 runs where the performance of the following multi-object filters/smoothers are compared:
\begin{enumerate}
  \item Track-oriented PMB filter \cite{pmbmpoint}, referred to as TO-PMB.
  \item Variational PMB filter \cite{variational}, referred to as V-PMB.
  \item Multi-trajectory particle smoother with TO-PMB filtering densities, referred to as BS-TO-PMB.
  \item Multi-trajectory particle smoother with V-PMB filtering densities, referred to as BS-V-PMB.
  \item Trajectory PMBM filter for the set of all trajectories \cite{granstrom2018poisson}, referred to as T-PMBM.
  \item Trajectory PMB filter for the set of all trajectories \cite{garcia2020trajectory}, referred to as T-PMB.
  \item $\delta$-GLMB filter with a multi-scan estimator \cite{nguyen2019glmb} and RTS smoothing, referred to as GLMB.
  \item Multi-scan GLMB with batch smoothing \cite{vo2019multi}, referred to as M-GLMB.
\end{enumerate}
We note that both 7) and 8) can be considered as implementations of the trajectory $\text{MBM}_{01}$ filter in \cite{garcia2019multiple}.

\begin{figure}[!t]
  \centering
  \includegraphics[width=\columnwidth]{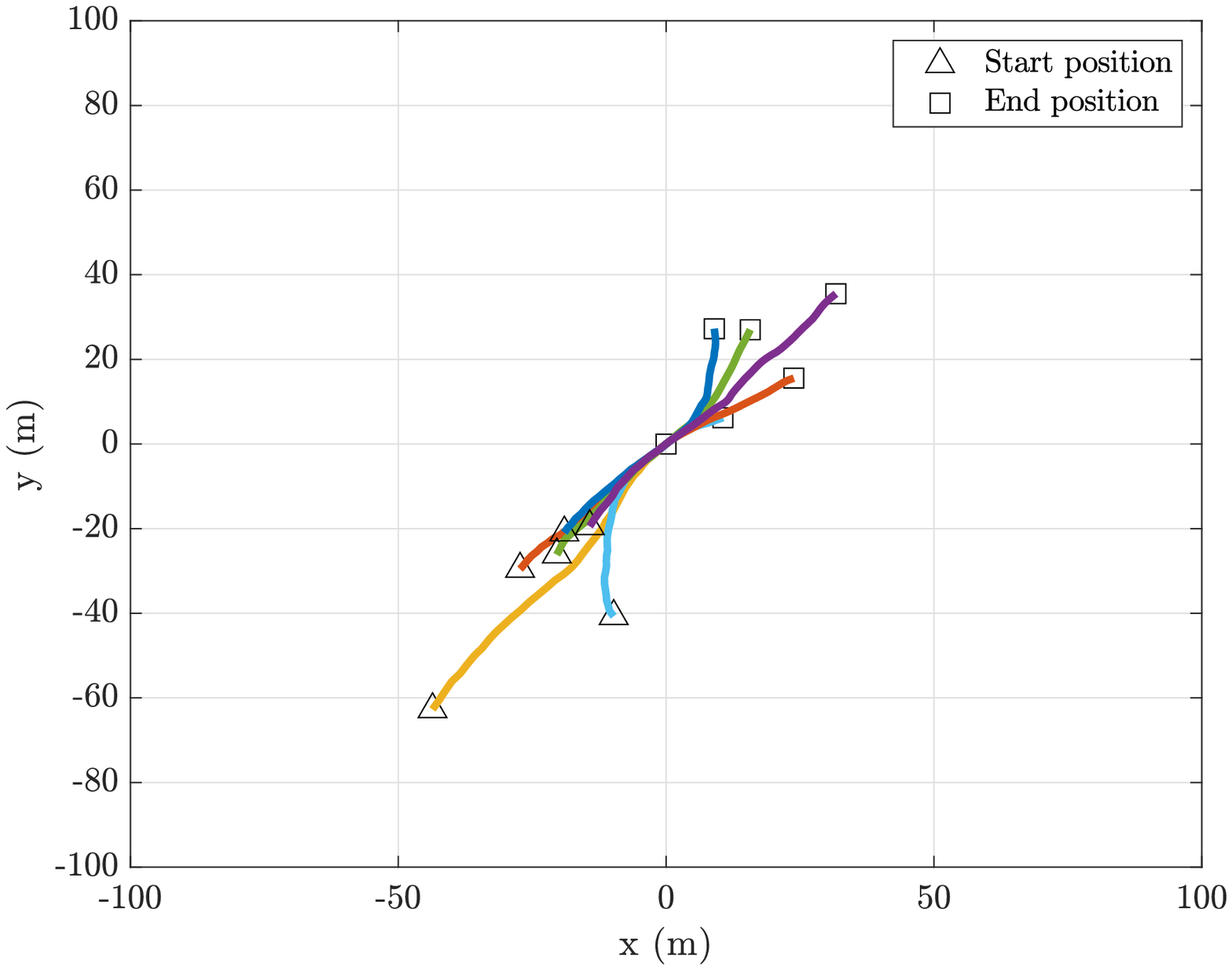}
  \includegraphics[width=\columnwidth]{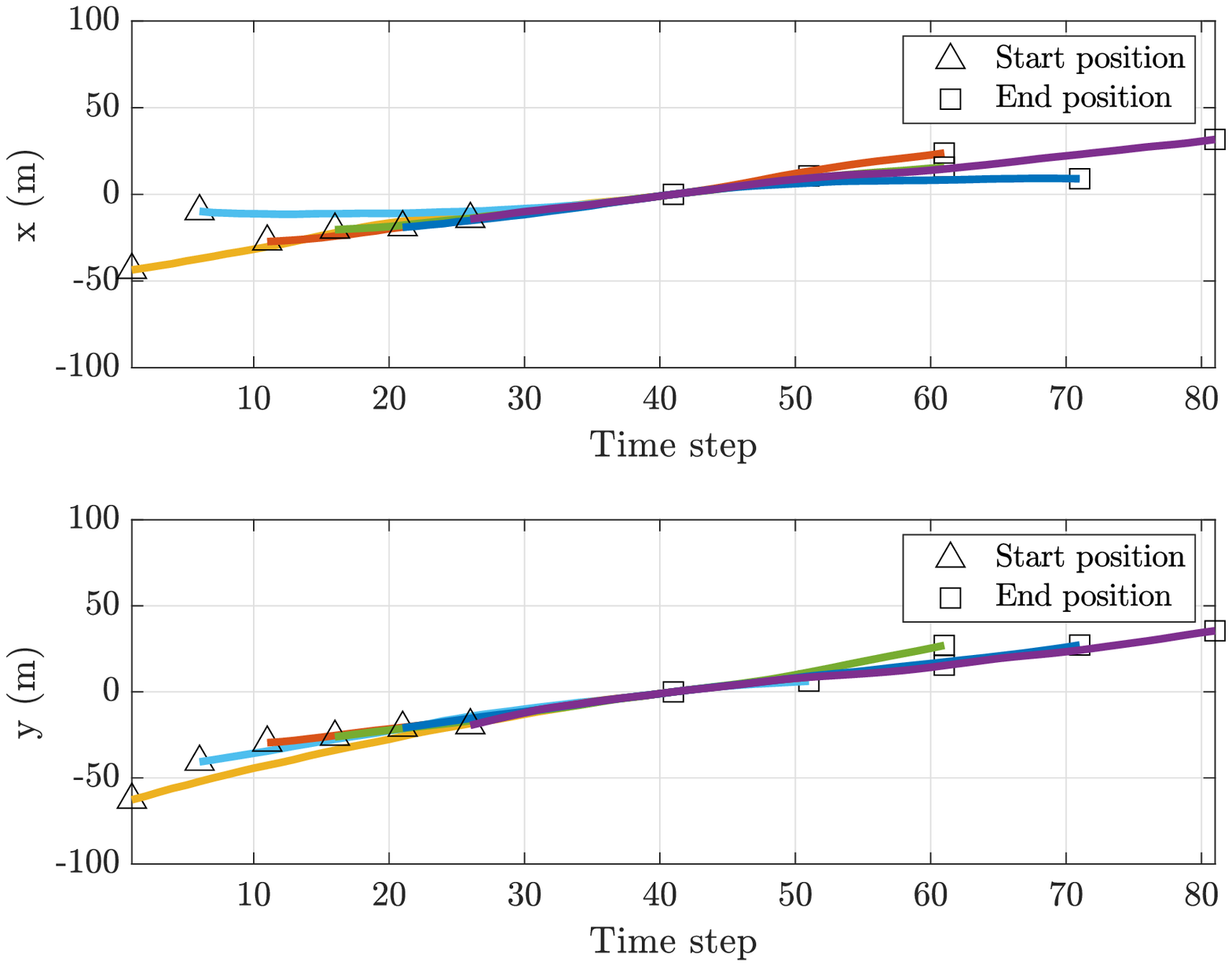}
  \caption{True trajectories of the scenario with a high risk of track coalescence for one of 500 Monte Carlo runs. The figure on the top shows the two-dimensional trajectories, and how their decompositions into $x$ and $y$ dimensions vary over time is illustrated in the two figures on the bottom. Note that for different Monte Carlo runs, only the birth and death time of objects are deterministic. Start/end positions of trajectories are marked by triangle/square, respectively. There are six objects. They are born at time step 1, 6, 11, 16, 21, 26 and die at time step 41, 51, 61, 61, 71, 81, respectively. At time step 41, one object dies when all objects are in close proximity.}
  \label{fig_gt}
\end{figure}

We consider a two-dimensional scenario with $81$ time steps where six initially well-separated objects move in close proximity to each other and thereafter separate, in the area $[-100\text{~m},100\text{~m}]\times[-100\text{~m},100\text{~m}]$. The true trajectories of the simulated scenario are illustrated in Fig. \ref{fig_gt}. We use a nearly constant velocity motion model with sampling period $T_s = 1~\text{s}$ and single-object state transition density parameterised by
\begin{equation*}
  F = I_2 \otimes \begin{bmatrix}
    1 & T_s\\
    0 & 1
  \end{bmatrix},\quad Q = \sigma_q^2I_2\otimes\begin{bmatrix}
    T_s^3/3 & T_s^2/2\\
    T_s^2/2 & T_s
  \end{bmatrix}
\end{equation*}
where $I_2$ is an identity matrix, $\otimes$ denotes the Kronecker product, and $\sigma_q = 0.1$. Each object survives with probability $p^S = 0.98$. We also consider the point object measurement model where each object generates at most one measurement. The probability of detection is $p^D = 0.7$ and the clutter is uniformly distributed in the tracking area with Poisson rate $\lambda^C = 30$. The measurement model is also linear and Gaussian with observation matrix $I_2\otimes\begin{bmatrix}
  1 & 0
\end{bmatrix}$ and measurement noise covariance $\sigma_r^2I_2$ where $\sigma_r = 1$.


For implementations with Poisson birth model, the Poisson birth intensity is a single Gaussian (c.f. \eqref{eq_Gaussian_birth}) with parameter $N_k^b = 1$, $w_k^{b,1}=0.05$, $m^{b,1}_k = [-25,1,-25,1]^T$ and $P^{b,1}_k = \text{diag}(225,1,225,1)$. The MB birth model used in GLMB and MS-GLMB contains a single Bernoulli with the same PHD as the Poisson birth process. 

All the implementations use ellipsoidal gating (with gating size computed using the inverse-chi-squared distribution at probability 0.9999) to remove unlikely local hypotheses and Murty's algorithm to find the $M$-best global hypotheses with highest weight, with the only exception being MS-GLMB where the multi-scan data association problem is solved using Gibbs sampling. The maximum number of hypotheses is 100 for TO-PMB, V-PMB and T-PMB, and 1000 for T-PMBM and GLMB. For filters with Poisson birth, Bernoulli components with probability of existence smaller than $10^{-4}$ and Poisson components with weights smaller than $10^{-4}$ are pruned. For T-PMBM and T-PMB, Bernoulli components with probability of being alive at the current time step smaller than $10^{-4}$ are considered dead, and both filters are implemented without $L$-scan approximation, i.e., single-object states at different time steps are not considered independent. For T-PMBM and GLMB, we prune global hypotheses with weight smaller than $10^{-4}$. For MS-GLMB, the smoothed multi-trajectory estimate obtained from GLMB is used to initialise the Markov chain and 1000 iterations are used in the multi-scan Gibbs sampler. For BS-TO-PMB and BS-V-PMB, the number of particles is 1000 and only a maximum of 100 global hypotheses with weight larger than $10^{-4}$ can be sampled.

For filters with Poisson birth, estimates are extracted from Bernoulli components with probability of existence $r \geq 0.5$. For T-PMBM, this is done for the global hypothesis with the highest weight. Also note that we only extract T-PMBM and T-PMB estimates at the last time step. For GLMB and M-GLMB, we also report estimates from the global hypothesis with the highest weight. For BS-TO-PMB and BS-V-PMB, we report the particle (set of trajectories estimate) with the highest likelihood accumulated over time, see Appendix \ref{pseudo_code} for implementation details. 

The multi-object state estimation performance is evaluated using the generalised optimal sub-pattern assignment (GOSPA) metric \cite{gospa} with $\alpha =2$, $c=20$ and $p=1$. Given a metric $d_b(\cdot,\cdot)$ in $\mathbb{R}^{n_x}$ , a scalar $c > 0$, and a scalar $p$ with $1 \leq p < \infty$, the GOSPA metric ($\alpha = 2$) between sets ${\bf x}$ and ${\bf y}$ is  \cite[Proposition 1]{gospa}
\begin{multline*}
  d({\bf x},{\bf y}) \\= \left[\min _{\theta \in \Gamma}\left(\sum_{(i, j) \in \theta} d_b^p\left(x_{i}, y_{j}\right)+\frac{c^{p}}{2}(|{\bf x}|+|{\bf y}|-2|\theta|)\right)\right]^{\frac{1}{p}}
\end{multline*}
where $\theta$ is an assignment set between sets $\{1,\dots,|{\bf x}|\}$ and $\{1,\dots,|{\bf y}|\}$, and $\Gamma$ is the set of all possible assignment sets. In addition, the multi-trajectory estimation performance is evaluated using the linear programming (LP) metric for sets of trajectories \cite{garcia2020metric} with $c=20$, $p=1$ and $\gamma=2$, which is an extension of GOSPA to sets of trajectories.

\begin{figure}[!t]
  \centering
  \includegraphics[width=\columnwidth]{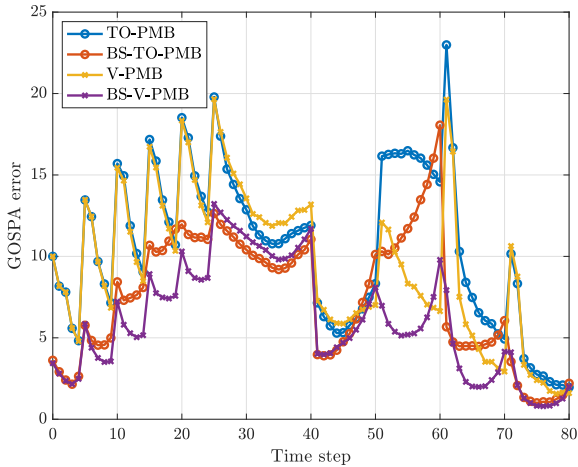}
  \caption{GOSPA error versus time for the scenario with a high risk of track coalescence. The two multi-trajectory particle smoothers outperform their corresponding multi-object filters in terms of GOSPA error at almost all the time steps.}
  \label{fig_gospa}
\end{figure}

\begin{table}[!t]
  \caption{Average GOSPA error and its decomposition where localisation refers to the localisation error normalised by the estimated cardinality}
  \label{table_gospa}
  \centering
  \begin{tabular}{ccccc}
  \hline
            & GOSPA & Localisation & Missed & False \\ \hline
  TO-PMB    & 859.5  & 93.6         & 358.5   & 196.0  \\
  BS-TO-PMB & 606.8  & 64.3         & 272.3   & 119.2  \\
  V-PMB     & 772.0  & 83.3         & 363.8   & 140.0  \\
  BS-V-PMB & $\mathbf{492.3}$  & 55.7         & 226.7   & 78.6  \\ \hline
  \end{tabular}
\end{table}

\begin{figure}[!t]
  \centering
  \includegraphics[width=\columnwidth]{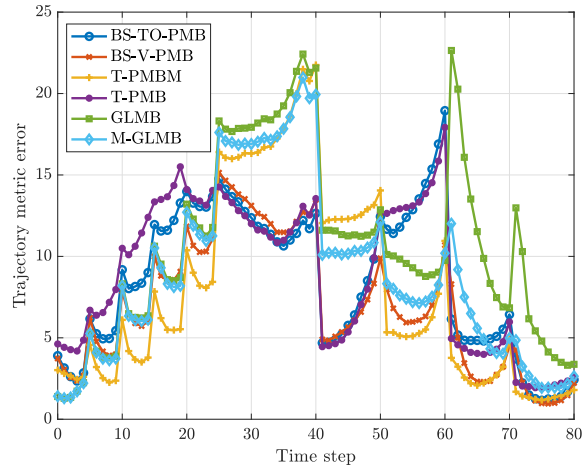}
  \caption{LP trajectory metric error versus time for the scenario with a high risk of track coalescence. BS-V-PMB has the best overall performance.}
  \label{fig_tra_metric}
\end{figure}

\begin{table}[!t]
  \caption{Average LP trajectory metric error and its decomposition}
  \label{table_tra_metric}
  \centering
  \begin{tabular}{cccccc}
  \hline
            & Total & Localisation & Missed & False & Switch \\ \hline
  BS-TO-PMB & 688.2              & 277.9         & 274.9   & 118.3 & 17.2   \\
  BS-V-PMB  & $\mathbf{564.4}$   & 242.6         & 225.2   & 79.0  & 17.7   \\
  T-PMBM    & 640.5              & 243.9         & 356.2   & 22.9  & 17.5   \\
  T-PMB     & 722.3              & 280.3         & 237.8   & 185.9 & 18.3   \\
  GLMB      & 867.0              & 285.9         & 329.8   & 232.1 & 19.2   \\
  M-GLMB    & 724.3              & 262.5         & 324.5   & 116.5  & 18.8   \\ \hline
  \end{tabular}
\end{table}

\begin{figure*}[!t]
  \centering
  \subfloat[\label{tra_loc}]
  {\includegraphics[width=0.5\columnwidth]{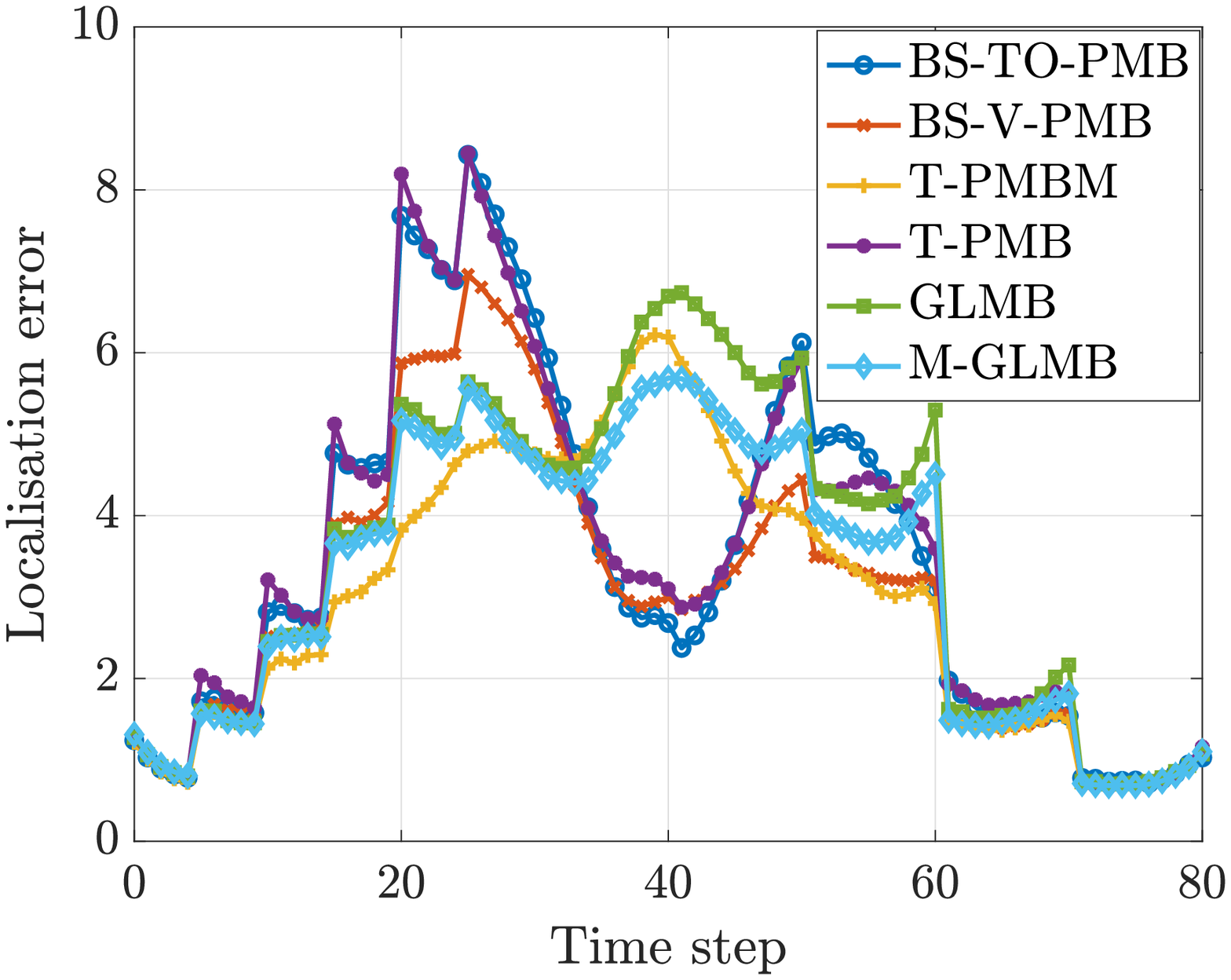}}
  \subfloat[\label{tra_missed}]
  {\includegraphics[width=0.5\columnwidth]{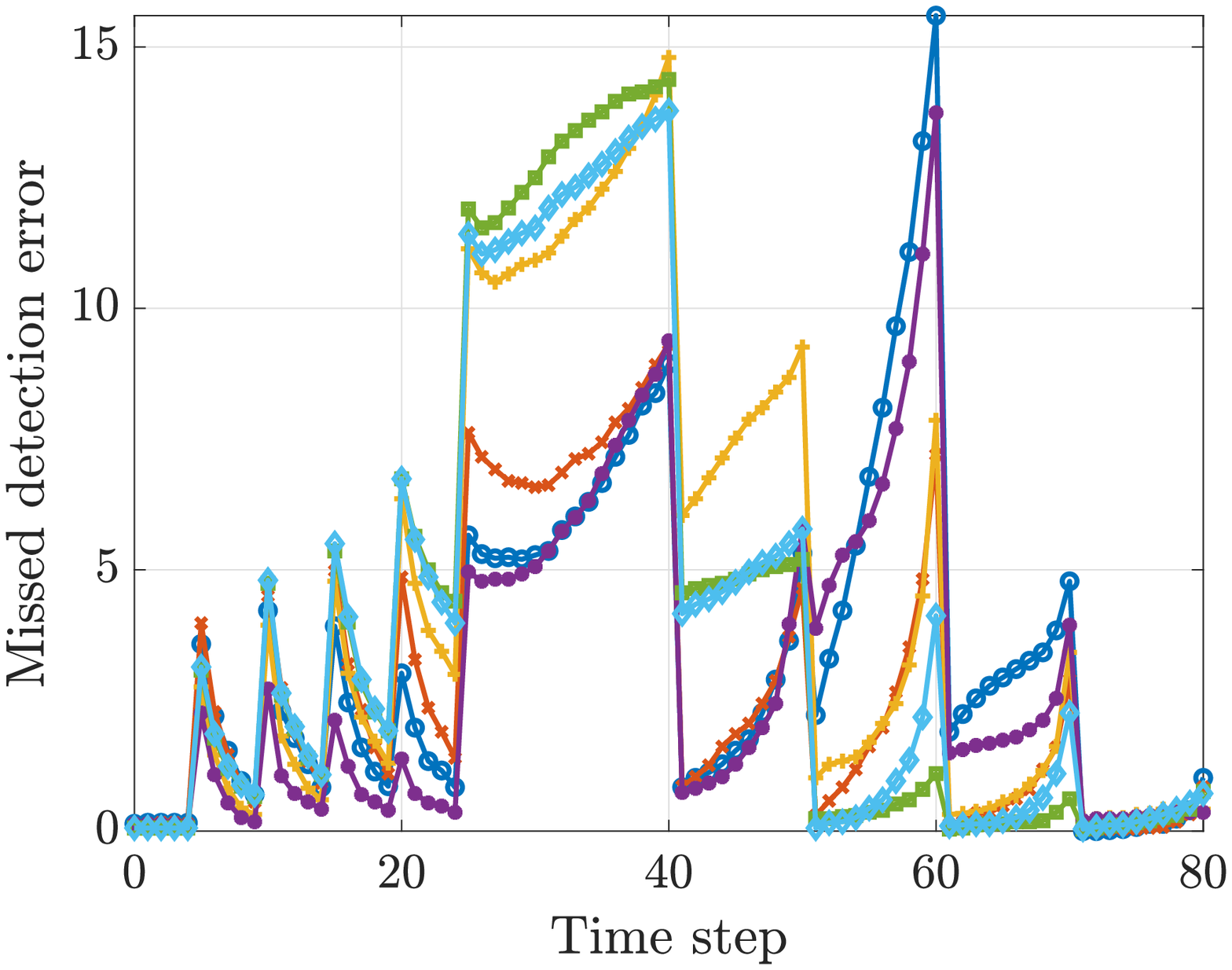}}
  \subfloat[\label{tra_false}]
  {\includegraphics[width=0.5\columnwidth]{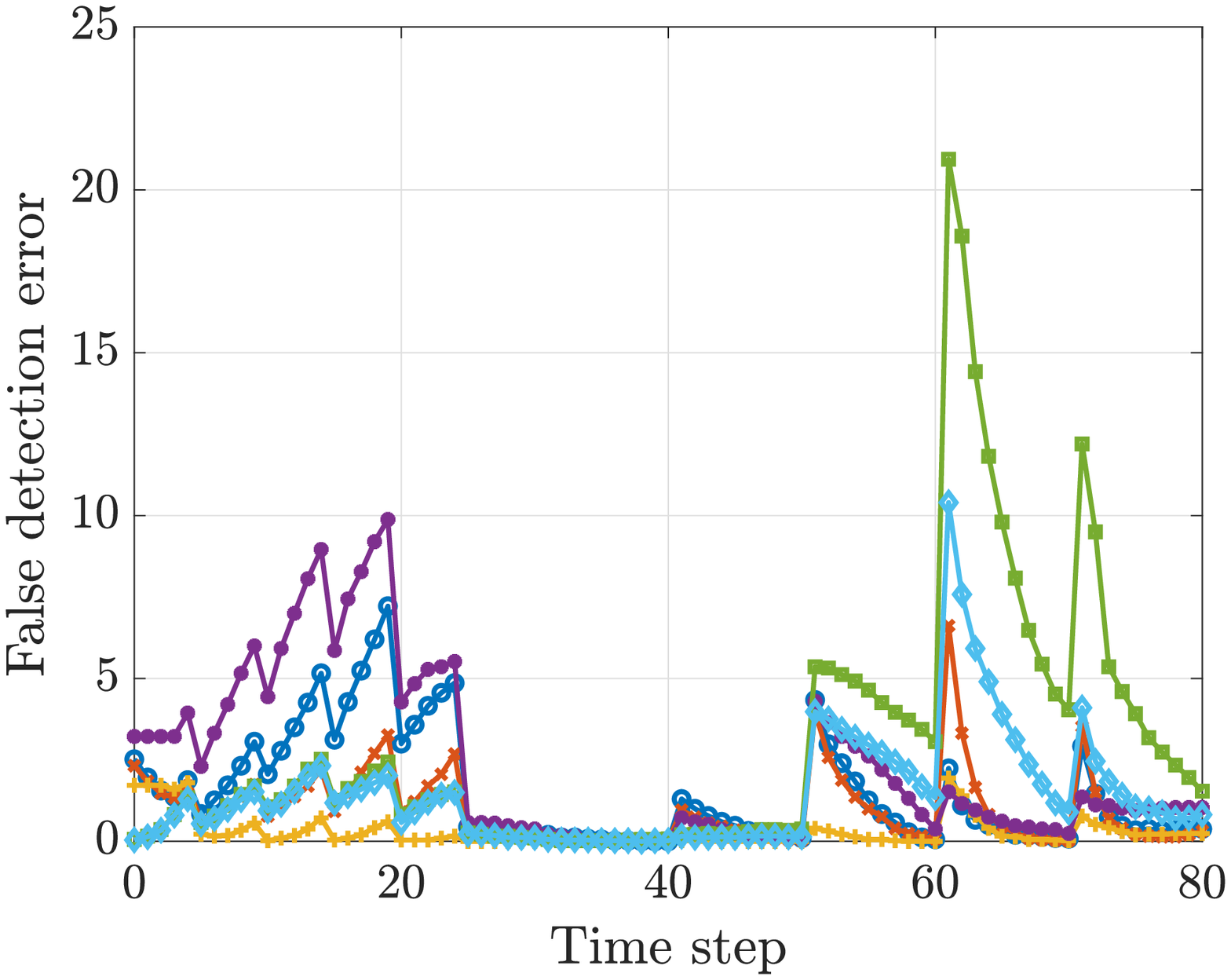}}
  \subfloat[\label{tra_switch}]
  {\includegraphics[width=0.5\columnwidth]{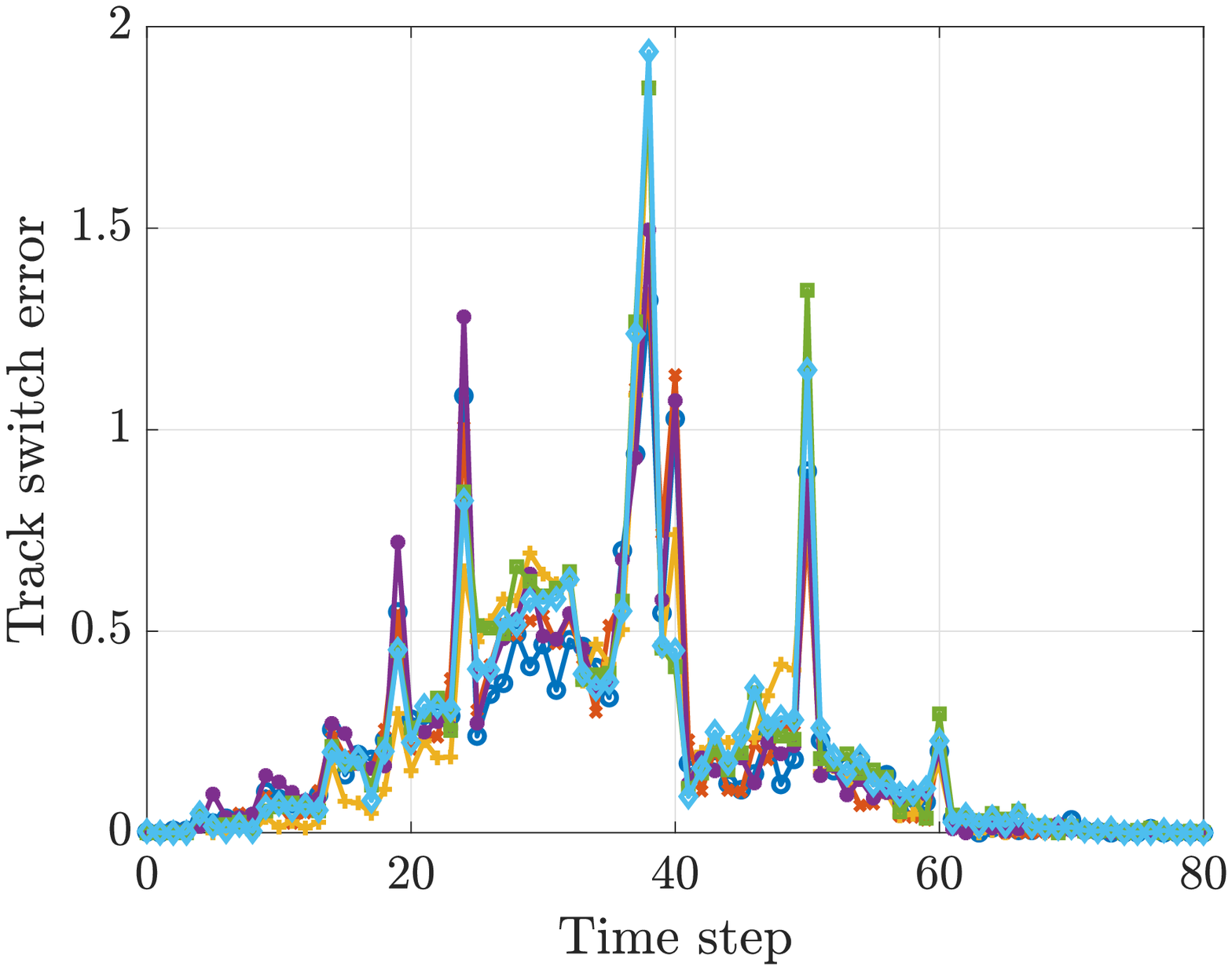}}
  \caption{Decompositions of trajectory metric error versus time for the scenario with a high risk of track coalescence.}
  \label{fig_tra_metric_decomposition}
\end{figure*}

\begin{figure*}[!t]
  \centering
  \subfloat[\label{stat_card}]
  {\includegraphics[width=0.666\columnwidth]{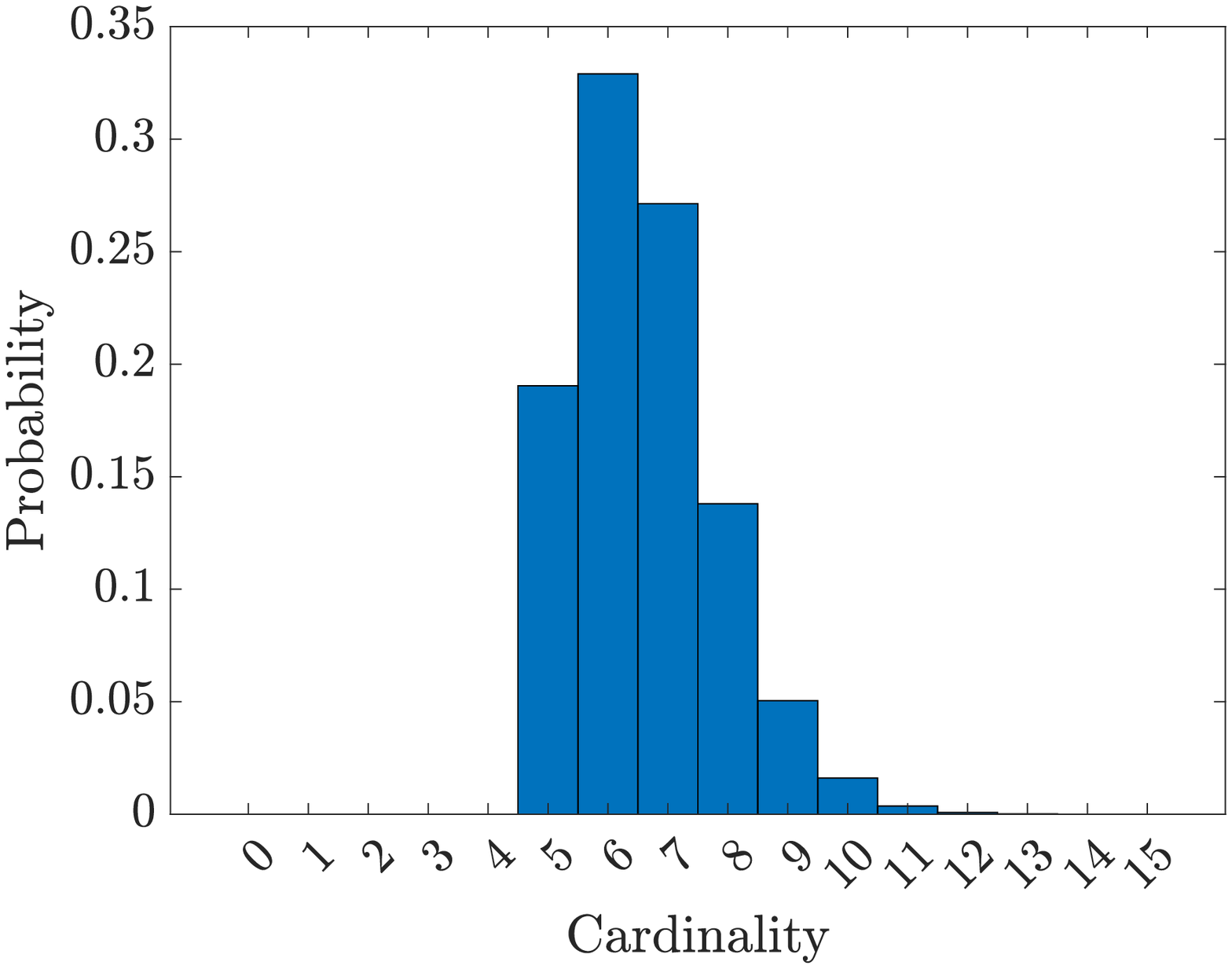}}
  \subfloat[\label{stat_birth}]
  {\includegraphics[width=0.666\columnwidth]{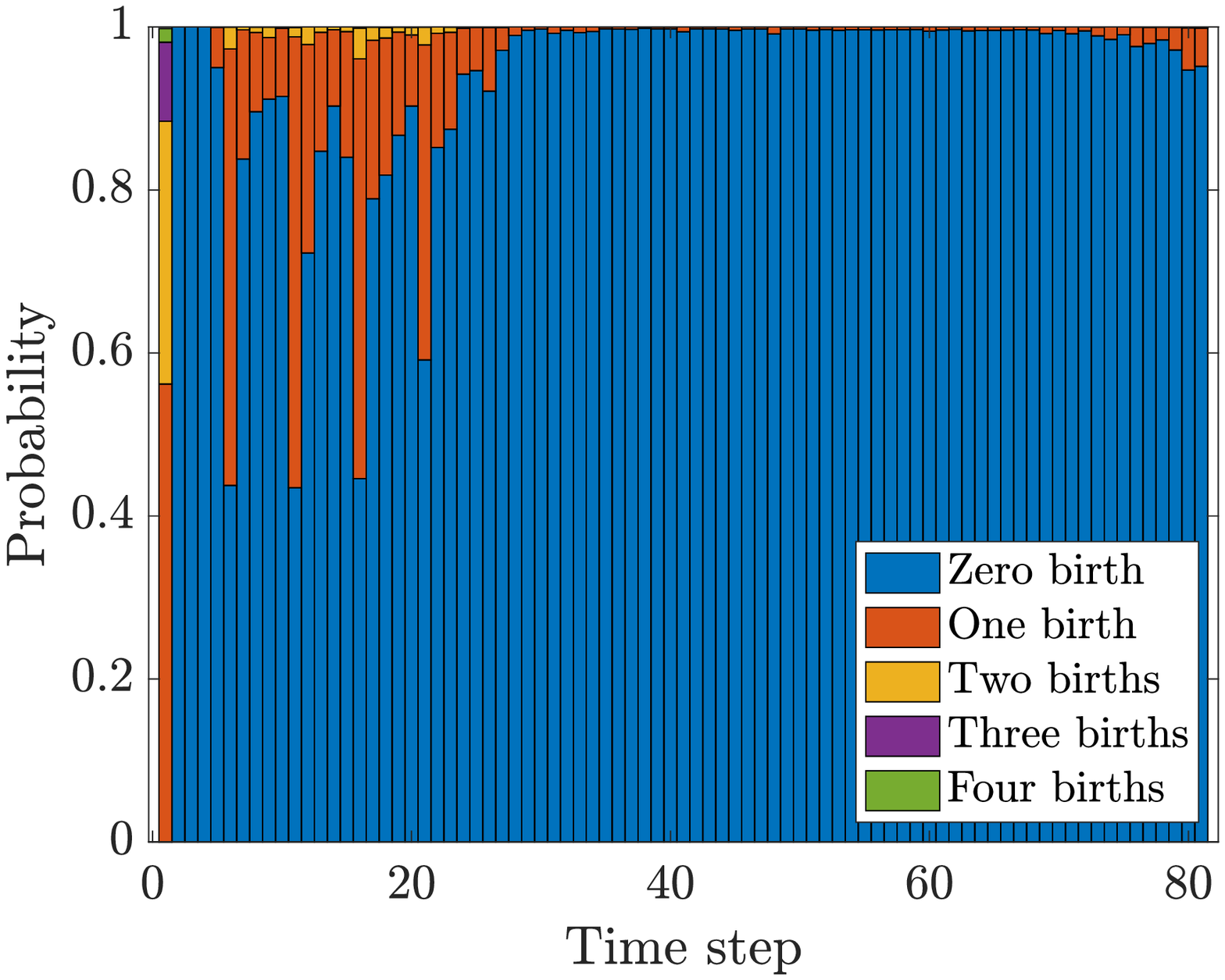}}
  \subfloat[\label{stat_death}]
  {\includegraphics[width=0.666\columnwidth]{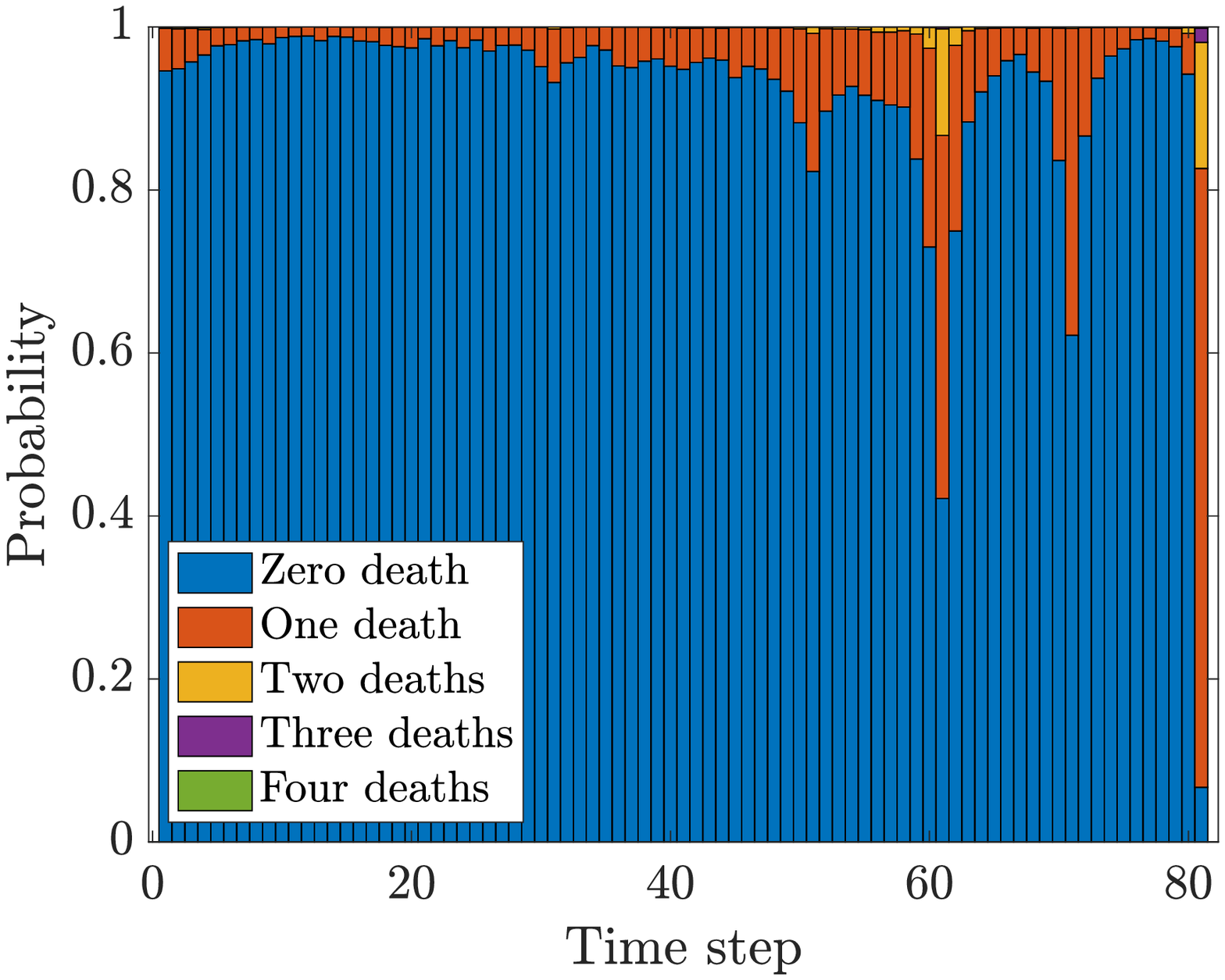}}
  \caption{Some statistics (averaged over 500 Monte Carlo runs) of the particle representation of the multi-trajectory density $\pi_{1:K|K}({\bf X})$ obtained from BS-V-PMB. The computation of these statistics is given in Appendix \ref{statistics}. Subfigure (a) shows the cardinality distribution of the set ${\bf X}$ of trajectories in the time interval $1:K$. Subfigure (b) and (c), respectively, show the cardinality distributions of trajectory birth and death at different time steps where the probability of a specific birth/death event is represented by the length of its corresponding bar.}
  \label{fig_stats}
\end{figure*}

We first analyse how the two multi-trajectory smoothers can improve multi-object estimation performance with respect to the forward filters. For TO-PMB, BS-TO-PMB, V-PMB and BS-V-PMB, the GOSPA error versus time is shown in Fig. \ref{fig_gospa}, and the average GOSPA error and its decomposition into localisation error, missed detection error and false detection error are presented in Table \ref{table_gospa}. It can be seen from Fig. \ref{fig_gospa} that TO-PMB shows the largest estimation error when objects moving in close proximity begin to separate, a problem known as track coalescence commonly observed in JPDAF. As a comparison, V-PMB resolves the coalescence by using a more accurate MB approximation method \cite{variational}. Both BS-TO-PMB and BS-V-PMB outperform their corresponding forward filters in terms of localisation error, missed and false detections by a large margin. Between these two smoothers, BS-V-PMB has better estimation performance than BS-TO-PMB.

We proceed to analyse the trajectory estimation performance of different implementations. The LP trajectory metric error versus time for all the implementations that estimate trajectories is shown in Fig. \ref{fig_tra_metric}, and the numerical values of the average LP trajectory metric error are presented in Table \ref{table_tra_metric}. On the whole, BS-V-PMB has the best estimation performance averaged over different time steps, followed by T-PMBM. In particular, BS-V-PMB has the best performance when objects are in close proximity, whereas T-PMBM has the best performance on initiating and terminating trajectories. In principle, T-PMBM will produce optimal trajectory estimates if it is implemented without approximation, and we apply an optimal estimator (e.g. in the sense of minimising the mean trajectory metric error). However, in practice, we use pruning and a suboptimal estimator. When objects are well-spaced, only hypotheses with negligible weight are pruned, so the performance of T-PMBM is effectively optimal. With closely-spaced objects, there are many feasible hypotheses which cannot be effectively enumerated, and the additional ability of BS-V-PMB to reason over the entire sequence is clearly evident. BS-TO-PMB outperforms T-PMB, with the latter being an efficient approximation of T-PMBM. GLMB is less accurate than the other implementations using Poisson birth model due to their less efficient representation of the multi-object posterior \cite{pmbmpoint2}, even though for the considered scenario where at most one object is born at a time it is more beneficial to use a Bernoulli birth model. M-GLMB has better performance than GLMB by improving the GLMB estimates using a multi-scan Gibbs sampler with batch smoothing.

The decompositions of the LP trajectory metric error into localisation error, missed detection error, false detection error and track switch error are shown in Fig. \ref{fig_tra_metric_decomposition}. T-PMBM, GLMB and M-GLMB show large missed detection errors when objects are in close proximity. In this case, many global hypotheses in T-PMBM and GLMB can have non-negligible weights due to the high data association uncertainty, and capping the number of global hypotheses may result in larger approximation errors than merging multiple global hypotheses into one. This explains why T-PMBM has worse performance than its approximation T-PMB when objects are in close proximity (before separation). Moreover, T-PMBM has very small false detection error, whereas GLMB has difficulty in terminating trajectories of dead objects. T-PMB has less missed detection error but more false detection error than BS-TO-PMB. The track switch error of all the implementations becomes large when objects are in close proximity, and implementations with Poisson birth model have lower track switch error than the two GLMB implementations.

Fig. \ref{fig_stats} shows the cardinality distributions of the estimated set of trajectories as well as the trajectory birth and death at different time steps, computed using the particle representations of the multi-trajectory density $\pi_{1:K|K}({\bf X})$ in BS-V-PMB. As can be seen, these statistics, in general, well reflect the ground truth except at time step 41 when one object died. Specifically, Fig. \ref{stat_birth} and Fig. \ref{stat_death} show that it is very likely that no object dies at time step 41. 

The average execution times in seconds of a single run\footnote{MATLAB implementation on 3.0 GHz Intel Core i5.} (81
time steps) for $p^D =0.7$, $\lambda^C =30$, $\sigma_q = 0.1$, and $\sigma_r=1$ are: 281.0 (BS-TO-PMB), 209.4 (BS-V-PMB), 171.0 (T-PMBM), 6.5 (T-PMB), 11.1 (GLMB), 8683.7 (M-GLMB). The fastest implementation is T-PMB, followed by GLMB. T-PMBM is slower than GLMB as we do not use the L-scan approximation \cite{garcia2020trajectory}. BS-V-PMB is faster than BS-TO-PMB even though V-PMB is slower than TO-PMB. This is due to the fact that the computational bottleneck of BS-V-PMB and BS-TO-PMB is backward simulation and that it is faster to run backward simulation on the PMB filtering densities obtained using V-PMB for the considered scenario. M-GLMB is significantly slower than the other implementations as it solves an 81-scan data association problem.


We proceed to analyse the performance of the filters and smoothers with Poisson birth for different scene parameters, and the results are shown in Table \ref{table_scene}. In general, BS-V-PMB has the best trajectory estimation performance, followed by T-PMBM. As expected, if the scenario has lower signal-to-noise ratio, e.g., when the motion noise or clutter intensity increases, performance of all the filters/smoothers decreases. If the scenario has higher signal-to-noise ratio, e.g., measurement noise or clutter intensity decreases, or probability of detection increases, performance of all the filters/smoothers increases. 

\begin{table}[!t]
  \caption{Average LP trajectory metric error for different scene parameters of the scenario with a high risk of track coalescence}
  \label{table_scene}
  \centering
  \begin{tabular}{ccccc}
  \hline
            & BS-TO-PMB & BS-V-PMB & T-PMBM & T-PMB \\ \hline
  No change    & 688.2  & $\mathbf{564.4}$  & 640.5   & 722.3  \\
  $\sigma_q=0.5$ & 1096.4  & $\mathbf{915.9}$  & 926.1   & 1088.2  \\
  $\sigma_r=0.5$ & 460.5  & $\mathbf{374.7}$ & 409.8   & 476.4  \\
  $p^D = 0.8$ & 559.0  & $\mathbf{457.2}$ & 505.9   & 582.6  \\ 
  $\lambda^C = 10$ & 642.2  & $\mathbf{511.7}$  & 565.9   & 668.3  \\
  $\lambda^C = 50$ & 712.1  & $\mathbf{599.5}$   & 693.2   & 752.2  \\\hline
  \end{tabular}
\end{table}

\begin{figure}[!t]
  \centering
  \includegraphics[width=\columnwidth]{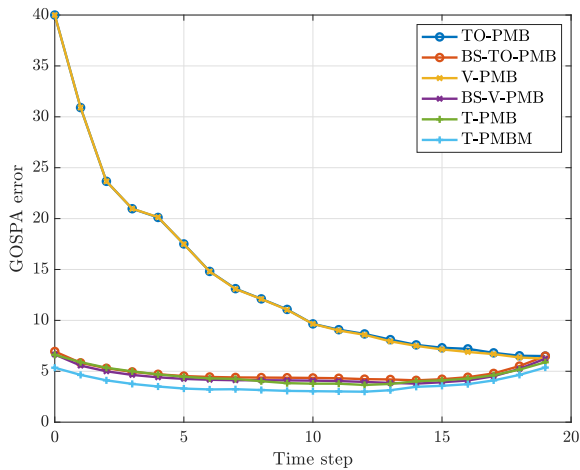}
  \caption{GOSPA error versus time for the scenario with simultaneous object births. The two trajectory filters and the two multi-trajectory particle smoothers significantly outperform the two PMB filters based on sets of objects.}
  \label{fig_gospa2}
\end{figure}

\begin{table}[!t]
  \caption{Average GOSPA error and its decomposition where localisation refers to the localisation error normalised by the estimated cardinality}
  \label{table_gospa2}
  \centering
  \begin{tabular}{ccccc}
  \hline
            & GOSPA & Localisation & Missed & False \\ \hline
  TO-PMB    & 281.7  & 24.1         & 192.1   & 11.6  \\
  BS-TO-PMB & 96.5  & 15.4         & 33.8   & 3.6  \\
  V-PMB     & 280.2  & 24.1         & 191.2   & 11.2  \\
  BS-V-PMB & 90.7  & 15.4         & 29.8   & 1.8  \\ 
  T-PMBM & $\mathbf{74.4}$  & 14.4         & 17.0   & 1.2  \\
  T-PMB & 91.5  & 15.7         & 22.8   & 7.7  \\\hline
  \end{tabular}
\end{table}

To further demonstrate the ability of the proposed multi-trajectory smoothers to infer object birth locations before first detection, we consider another scenario with $20$ time steps where four objects are born at time step 1 and no object dies. Compared to the first scenario, here the probability of detection is $p^D = 0.5$ and the Poisson clutter rate is $\lambda^C = 5$. We compare the performance of different implementations with Poisson birth model. The Poisson birth intensity is parameterised by $N_k^b = 1$, $w_k^{b,1}=0.02$, $x^{b,1}_k = [0,0,0,0]^T$, $P^{b,1}_k = \text{diag}(100,2,100,2)^2$, and the initial Poisson intensity for undetected objects is set to $\lambda^u_{0|0}(x) = 4{\cal N}(x;x^{b,1}_k,P^{b,1}_k)$.

The GOSPA error versus time for the considered scenario is shown in Fig. \ref{fig_gospa2} and its decomposition is presented in Table \ref{table_gospa2}. The results show that T-PMBM has the best estimation performance, and that TO-PMB and V-PMB have the worst estimation performance. The estimation performance of T-PMB, BS-TO-PMB and BS-V-PMB is similar, and it is slightly worse than T-PMBM. Due to the low detection probability, objects were usually detected a few time steps after they were born. This explains the high missed detection error of TO-PMB and V-PMB. The estimation of object states before first detection can be obtained by considering the posterior density on sets of trajectories, which captures all the information about the trajectories, including those of undetected objects. For T-PMBM and T-PMB, this information is explicitly carried over time via forward filtering, whereas for BS-TO-PMB and BS-V-PMB, this information is inferred from filtering densities of detected objects at later time steps and gradually recovered via backward smoothing.

At last, we note that the forward-backward smoothers BS-TO-PMB, BS-V-PMB and the trajectory filters T-PMBM, T-PMB are suitable for different applications. BS-TO-PMB and BS-V-PMB are offline methods and their backward smoothing steps do not utilise any information of the measurements and the multi-object measurement model, and therefore they are more suitable for offline trajectory analytics. T-PMBM and T-PMB are online methods, and they can also be applied to batch problems. The difference is that T-PMBM and T-PMB solve the data associations while filtering, whereas with backward simulation, we are not constrained to the previously solved data associations in the sense that future measurements can be utilised to improve trajectory estimation.

\section{Conclusions}

In this paper, we have derived a multi-trajectory backward smoothing equation based on sets of trajectories. This allows us to leverage filters that do not keep trajectory information to compute the posterior density of sets of trajectories, and has important applications to offline trajectory analytics. In addition, we have proposed a multi-trajectory particle smoother using backward simulation for PMB filtering densities along with its tractable implementation based on ranked assignment. The simulation results show that the proposed methods have superior trajectory estimation performance compared to several state-of-the-art algorithms.

A follow-up work direction is developing an implementation that works for forward densities with particle representation. In addition, it would be interesting to study how to extract better estimates from the particle representation of multi-trajectory densities \eqref{eq_particle}, e.g., by merging different particles in the multi-object trajectory space.

\bibliographystyle{IEEEtran}
\bibliography{mybibli.bib}

\begin{thebibliography}{10}
\providecommand{\url}[1]{#1}
\csname url@samestyle\endcsname
\providecommand{\newblock}{\relax}
\providecommand{\bibinfo}[2]{#2}
\providecommand{\BIBentrySTDinterwordspacing}{\spaceskip=0pt\relax}
\providecommand{\BIBentryALTinterwordstretchfactor}{4}
\providecommand{\BIBentryALTinterwordspacing}{\spaceskip=\fontdimen2\font plus
\BIBentryALTinterwordstretchfactor\fontdimen3\font minus
  \fontdimen4\font\relax}
\providecommand{\BIBforeignlanguage}[2]{{%
\expandafter\ifx\csname l@#1\endcsname\relax
\typeout{** WARNING: IEEEtran.bst: No hyphenation pattern has been}%
\typeout{** loaded for the language `#1'. Using the pattern for}%
\typeout{** the default language instead.}%
\else
\language=\csname l@#1\endcsname
\fi
#2}}
\providecommand{\BIBdecl}{\relax}
\BIBdecl

\bibitem{bar2004estimation}
Y.~Bar-Shalom, X.~R. Li, and T.~Kirubarajan, \emph{Estimation with applications
  to tracking and navigation: theory algorithms and software}.\hskip 1em plus
  0.5em minus 0.4em\relax John Wiley \& Sons, 2004.

\bibitem{challa2011fundamentals}
S.~Challa, M.~R. Morelande, D.~Mu{\v{s}}icki, and R.~J. Evans,
  \emph{Fundamentals of object tracking}.\hskip 1em plus 0.5em minus
  0.4em\relax Cambridge University Press, 2011.

\bibitem{meyer2018message}
F.~Meyer, T.~Kropfreiter, J.~L. Williams, R.~Lau, F.~Hlawatsch, P.~Braca, and
  M.~Z. Win, ``Message passing algorithms for scalable multitarget tracking,''
  \emph{Proceedings of the IEEE}, vol. 106, no.~2, pp. 221--259, 2018.

\bibitem{streit2021analytic}
R.~L. Streit, R.~B. Angle, and M.~Efe, \emph{Analytic combinatorics in multiple
  object tracking}.\hskip 1em plus 0.5em minus 0.4em\relax Springer, 2021.

\bibitem{bar2009probabilistic}
Y.~Bar-Shalom, F.~Daum, and J.~Huang, ``The probabilistic data association
  filter,'' \emph{IEEE Control Systems Magazine}, vol.~29, no.~6, pp. 82--100,
  2009.

\bibitem{blackman2004multiple}
S.~S. Blackman, ``Multiple hypothesis tracking for multiple target tracking,''
  \emph{IEEE Aerospace and Electronic Systems Magazine}, vol.~19, no.~1, pp.
  5--18, 2004.

\bibitem{chong2019forty}
C.~Chong, S.~Mori, and D.~B. Reid, ``Forty years of multiple hypothesis
  tracking,'' \emph{Journal of Advances in Information Fusion}, vol.~14, no.~2,
  pp. 131--151, 2019.

\bibitem{mahler2007statistical}
R.~P. Mahler, \emph{Statistical Multisource-Multitarget Information
  Fusion}.\hskip 1em plus 0.5em minus 0.4em\relax Artech House Norwood, MA,
  2007.

\bibitem{mahler2003multitarget}
------, ``Multitarget {B}ayes filtering via first-order multitarget moments,''
  \emph{IEEE Transactions on Aerospace and Electronic systems}, vol.~39, no.~4,
  pp. 1152--1178, 2003.

\bibitem{mahler2007phd}
------, ``{PHD} filters of higher order in target number,'' \emph{IEEE
  Transactions on Aerospace and Electronic systems}, vol.~43, no.~4, pp.
  1523--1543, 2007.

\bibitem{pmbmpoint}
J.~L. Williams, ``Marginal multi-{B}ernoulli filters: {RFS} derivation of
  {MHT}, {JIPDA}, and association-based member,'' \emph{IEEE Transactions on
  Aerospace and Electronic Systems}, vol.~51, no.~3, pp. 1664--1687, 2015.

\bibitem{glmbconjugateprior}
B.-T. Vo and B.-N. Vo, ``Labeled random finite sets and multi-object conjugate
  priors,'' \emph{IEEE Transactions on Signal Processing}, vol.~61, no.~13, pp.
  3460--3475, 2013.

\bibitem{garcia2013two}
{\'A}.~F. Garc{\'\i}a-Fern{\'a}ndez, J.~Grajal, and M.~R. Morelande,
  ``Two-layer particle filter for multiple target detection and tracking,''
  \emph{IEEE Transactions on Aerospace and Electronic Systems}, vol.~49, no.~3,
  pp. 1569--1588, 2013.

\bibitem{garci2014bayesian}
{\'A}.~F. Garc{\'\i}a-Fern{\'a}ndez, M.~R. Morelande, and J.~Grajal, ``Bayesian
  sequential track formation,'' \emph{IEEE Transactions on Signal Processing},
  vol.~62, no.~24, pp. 6366--6379, 2014.

\bibitem{streit2018analytic}
R.~Streit, ``Analytic combinatorics and labeling in high level fusion and
  multihypothesis tracking,'' in \emph{21st International Conference on
  Information Fusion (FUSION)}.\hskip 1em plus 0.5em minus 0.4em\relax IEEE,
  2018, pp. 1--5.

\bibitem{aoki2016labeling}
E.~H. Aoki, P.~K. Mandal, L.~Svensson, Y.~Boers, and A.~Bagchi, ``Labeling
  uncertainty in multitarget tracking,'' \emph{IEEE Transactions on Aerospace
  and Electronic systems}, vol.~52, no.~3, pp. 1006--1020, 2016.

\bibitem{garcia2019multiple}
{\'A}.~F. Garc{\'\i}a-Fern{\'a}ndez, L.~{Svensson}, and M.~R. {Morelande},
  ``Multiple target tracking based on sets of trajectories,'' \emph{IEEE
  Transactions on Aerospace and Electronic Systems}, vol.~56, no.~3, pp.
  1685--1707, 2019.

\bibitem{granstrom2018poisson}
K.~Granstr{\"o}m, L.~Svensson, Y.~Xia, J.~Williams, and {\'A}.~F.
  Garc{\'\i}a-Fem{\'a}ndez, ``Poisson multi-{B}ernoulli mixture trackers:
  {C}ontinuity through random finite sets of trajectories,'' in \emph{21st
  International Conference on Information Fusion (FUSION)}.\hskip 1em plus
  0.5em minus 0.4em\relax IEEE, 2018, pp. 1--8.

\bibitem{garcia2019trajectory}
{\'A}.~F. Garc{\'\i}a-Fern{\'a}ndez and L.~Svensson, ``Trajectory {PHD} and
  {CPHD} filters,'' \emph{IEEE Transactions on Signal Processing}, vol.~67,
  no.~22, pp. 5702--5714, 2019.

\bibitem{xia2019multi}
Y.~Xia, K.~Granstr{\"o}m, L.~Svensson, {\'A}.~F. Garc{\'\i}a-Fern{\'a}ndez, and
  J.~L. Williams, ``Multi-scan implementation of the trajectory {P}oisson
  multi-{B}ernoulli mixture filter,'' \emph{Journal of Advances in Information
  Fusion}, vol.~14, no.~2, pp. 213--235, 2019.

\bibitem{garcia2020trajectory}
{\'A}.~F. Garc{\'\i}a-Fern{\'a}ndez, L.~Svensson, J.~L. Williams, Y.~Xia, and
  K.~Granstr{\"o}m, ``Trajectory {P}oisson multi-{B}ernoulli filters,''
  \emph{IEEE Transactions on Signal Processing}, vol.~68, pp. 4933--4945, 2020.

\bibitem{garcia2020trajectory3}
------, ``Trajectory multi-{B}ernoulli filters for multi-target tracking based
  on sets of trajectories,'' in \emph{IEEE 23rd International Conference on
  Information Fusion (FUSION)}.\hskip 1em plus 0.5em minus 0.4em\relax IEEE,
  2020, pp. 1--8.

\bibitem{vo2019multi}
B.-N. Vo and B.-T. Vo, ``A multi-scan labeled random finite set model for
  multi-object state estimation,'' \emph{IEEE Transactions on Signal
  Processing}, vol.~67, no.~19, pp. 4948--4963, 2019.

\bibitem{sarkka2013bayesian}
S.~S{\"a}rkk{\"a}, \emph{Bayesian filtering and smoothing}.\hskip 1em plus
  0.5em minus 0.4em\relax Cambridge University Press, 2013, no.~3.

\bibitem{kitagawa1994two}
G.~Kitagawa, ``The two-filter formula for smoothing and an implementation of
  the {G}aussian-sum smoother,'' \emph{Annals of the Institute of Statistical
  Mathematics}, vol.~46, no.~4, pp. 605--623, 1994.

\bibitem{lee2015smoothing}
D.~J. Lee and M.~E. Campbell, ``Smoothing algorithm for nonlinear systems using
  {G}aussian mixture models,'' \emph{Journal of Guidance, Control, and
  Dynamics}, vol.~38, no.~8, pp. 1438--1451, 2015.

\bibitem{rahmathullah2014merging}
A.~S. Rahmathullah, L.~Svensson, and D.~Svensson, ``Merging-based
  forward-backward smoothing on {G}aussian mixtures,'' in \emph{17th
  International Conference on Information Fusion (FUSION)}.\hskip 1em plus
  0.5em minus 0.4em\relax IEEE, 2014, pp. 1--8.

\bibitem{rahmathullah2014two}
------, ``Two-filter {G}aussian mixture smoothing with posterior pruning,'' in
  \emph{17th International Conference on Information Fusion (FUSION)}.\hskip
  1em plus 0.5em minus 0.4em\relax IEEE, 2014, pp. 1--8.

\bibitem{balenzuela2018accurate}
M.~P. Balenzuela, J.~Dahlin, N.~Bartlett, A.~G. Wills, C.~Renton, and
  B.~Ninness, ``Accurate {G}aussian mixture model smoothing using a two-filter
  approach,'' in \emph{IEEE Conference on Decision and Control (CDC)}.\hskip
  1em plus 0.5em minus 0.4em\relax IEEE, 2018, pp. 694--699.

\bibitem{chakravorty2006augmented}
R.~Chakravorty and S.~Challa, ``Augmented state integrated probabilistic data
  association smoothing for automatic track initiation in clutter,'' \emph{J.
  Adv. Inf. Fusion}, vol.~1, no.~1, pp. 63--74, 2006.

\bibitem{song2012smoothing}
T.~L. Song and D.~Mu{\v{s}}icki, ``Smoothing innovations and data association
  with {IPDA},'' \emph{Automatica}, vol.~48, no.~7, pp. 1324--1329, 2012.

\bibitem{muvsicki2013smoothing}
D.~Mu{\v{s}}icki, T.~L. Song, and T.~H. Kim, ``Smoothing multi-scan target
  tracking in clutter,'' \emph{IEEE Transactions on Signal Processing},
  vol.~61, no.~19, pp. 4740--4752, 2013.

\bibitem{clark2009joint}
D.~Clark, ``Joint target-detection and tracking smoothers,'' in \emph{Signal
  Processing, Sensor Fusion, and Target Recognition XVIII}, vol. 7336.\hskip
  1em plus 0.5em minus 0.4em\relax International Society for Optics and
  Photonics, 2009, p. 73360G.

\bibitem{vo2011bernoulli}
B.-T. Vo, D.~Clark, B.-N. Vo, and B.~Ristic, ``Bernoulli forward-backward
  smoothing for joint target detection and tracking,'' \emph{IEEE Transactions
  on Signal Processing}, vol.~59, no.~9, pp. 4473--4477, 2011.

\bibitem{vo2011closed}
B.-N. Vo, B.-T. Vo, and R.~P. Mahler, ``Closed-form solutions to
  forward-backward smoothing,'' \emph{IEEE Transactions on Signal Processing},
  vol.~60, no.~1, pp. 2--17, 2011.

\bibitem{rauch1965maximum}
H.~E. Rauch, F.~Tung, and C.~T. Striebel, ``Maximum likelihood estimates of
  linear dynamic systems,'' \emph{AIAA journal}, vol.~3, no.~8, pp. 1445--1450,
  1965.

\bibitem{mahalanabis1990improved}
A.~Mahalanabis, B.~Zhou, and N.~Bose, ``Improved multi-target tracking in
  clutter by {PDA} smoothing,'' \emph{IEEE Transactions on Aerospace and
  Electronic Systems}, vol.~26, no.~1, pp. 113--121, 1990.

\bibitem{koch2000fixed}
W.~Koch, ``Fixed-interval retrodiction approach to {B}ayesian {IMM-MHT} for
  maneuvering multiple targets,'' \emph{IEEE Transactions on Aerospace and
  Electronic Systems}, vol.~36, no.~1, pp. 2--14, 2000.

\bibitem{chen2001tracking}
B.~Chen and J.~K. Tugnait, ``Tracking of multiple maneuvering targets in
  clutter using {IMM/JPDA} filtering and fixed-lag smoothing,''
  \emph{Automatica}, vol.~37, no.~2, pp. 239--249, 2001.

\bibitem{rahmathullah2016batch}
A.~S. Rahmathullah, R.~Selvan, and L.~Svensson, ``A batch algorithm for
  estimating trajectories of point targets using expectation maximization,''
  \emph{IEEE Transactions on Signal Processing}, vol.~64, no.~18, pp.
  4792--4804, 2016.

\bibitem{briers2010smoothing}
M.~Briers, A.~Doucet, and S.~Maskell, ``Smoothing algorithms for state-space
  models,'' \emph{Annals of the Institute of Statistical Mathematics}, vol.~62,
  no.~1, p.~61, 2010.

\bibitem{sjpda}
L.~Svensson, D.~Svensson, M.~Guerriero, and P.~Willett, ``Set {JPDA} filter for
  multitarget tracking,'' \emph{IEEE Transactions on Signal Processing},
  vol.~59, no.~10, pp. 4677--4691, 2011.

\bibitem{variational}
J.~L. Williams, ``An efficient, variational approximation of the best fitting
  multi-{B}ernoulli filter,'' \emph{IEEE Transactions on Signal Processing},
  vol.~63, no.~1, pp. 258--273, 2015.

\bibitem{svensson2011multitarget}
D.~Svensson, L.~Svensson, M.~Guerriero, D.~F. Crouse, and P.~Willett, ``The
  multitarget set {JPDA} filter with target identity,'' in \emph{Signal
  Processing, Sensor Fusion, and Target Recognition XX}, vol. 8050.\hskip 1em
  plus 0.5em minus 0.4em\relax International Society for Optics and Photonics,
  2011, p. 805010.

\bibitem{nandakumaran2007improved}
N.~Nandakumaran, K.~Punithakumar, and T.~Kirubarajan, ``Improved multitarget
  tracking using probability hypothesis density smoothing,'' in \emph{Signal
  and Data Processing of Small Targets}, vol. 6699.\hskip 1em plus 0.5em minus
  0.4em\relax International Society for Optics and Photonics, 2007, p. 66990M.

\bibitem{clark2010first}
D.~E. Clark, ``First-moment multi-object forward-backward smoothing,'' in
  \emph{13th International Conference on Information Fusion}.\hskip 1em plus
  0.5em minus 0.4em\relax IEEE, 2010, pp. 1--6.

\bibitem{nadarajah2011multitarget}
N.~Nadarajah, T.~Kirubarajan, T.~Lang, M.~McDonald, and K.~Punithakumar,
  ``Multitarget tracking using probability hypothesis density smoothing,''
  \emph{IEEE Transactions on Aerospace and Electronic Systems}, vol.~47, no.~4,
  pp. 2344--2360, 2011.

\bibitem{mahler2012forward}
R.~P. Mahler, B.-T. Vo, and B.-N. Vo, ``Forward-backward probability hypothesis
  density smoothing,'' \emph{IEEE Transactions on Aerospace and Electronic
  Systems}, vol.~48, no.~1, pp. 707--728, 2012.

\bibitem{feng2016adaptive}
P.~Feng, W.~Wang, S.~M. Naqvi, and J.~Chambers, ``Adaptive retrodiction
  particle {PHD} filter for multiple human tracking,'' \emph{IEEE Signal
  Processing Letters}, vol.~23, no.~11, pp. 1592--1596, 2016.

\bibitem{li2016multi}
D.~Li, C.~Hou, and D.~Yi, ``Multi-{B}ernoulli smoother for multi-target
  tracking,'' \emph{Aerospace Science and Technology}, vol.~48, pp. 234--245,
  2016.

\bibitem{nagappa2017tractable}
S.~Nagappa, E.~D. Delande, D.~E. Clark, and J.~Houssineau, ``A tractable
  forward-backward {CPHD} smoother,'' \emph{IEEE Transactions on Aerospace and
  Electronic Systems}, vol.~53, no.~1, pp. 201--217, 2017.

\bibitem{streit2017interval}
R.~Streit, ``Interval/smoothing filters for multiple object tracking via
  analytic combinatorics,'' in \emph{20th International Conference on
  Information Fusion (Fusion)}.\hskip 1em plus 0.5em minus 0.4em\relax IEEE,
  2017, pp. 1--8.

\bibitem{beard2016generalised}
M.~Beard, B.~T. Vo, and B.-N. Vo, ``Generalised labelled multi-{B}ernoulli
  forward-backward smoothing,'' in \emph{19th International Conference on
  Information Fusion (FUSION)}.\hskip 1em plus 0.5em minus 0.4em\relax IEEE,
  2016, pp. 688--694.

\bibitem{liu2019computationally}
R.~Liu, H.~Fan, T.~Li, and H.~Xiao, ``A computationally efficient labeled
  multi-{B}ernoulli smoother for multi-target tracking,'' \emph{Sensors},
  vol.~19, no.~19, p. 4226, 2019.

\bibitem{liang2020multitarget}
G.~Liang, Q.~Li, B.~Qi, and L.~Qiu, ``Multitarget tracking using one time step
  lagged delta-generalized labeled multi-{B}ernoulli smoothing,'' \emph{IEEE
  Access}, vol.~8, pp. 28\,242--28\,256, 2020.

\bibitem{xia2020backward}
Y.~Xia, L.~Svensson, {\'A}.~F. Garc{\'\i}a-Fern{\'a}ndez, K.~Granstr{\"o}m, and
  J.~L. Williams, ``Backward simulation for sets of trajectories,'' in
  \emph{IEEE 23rd International Conference on Information Fusion
  (FUSION)}.\hskip 1em plus 0.5em minus 0.4em\relax IEEE, 2020, pp. 1--8.

\bibitem{lindsten2013backward}
F.~Lindsten and T.~B. Sch{\"o}n, ``Backward simulation methods for {M}onte
  {C}arlo statistical inference,'' \emph{Foundations and
  Trends{\textregistered} in Machine Learning}, vol.~6, no.~1, pp. 1--143,
  2013.

\bibitem{nguyen2019glmb}
T.~T.~D. Nguyen and D.~Y. Kim, ``{GLMB} tracker with partial smoothing,''
  \emph{Sensors}, vol.~19, no.~20, p. 4419, 2019.

\bibitem{granstrom2019poisson}
K.~Granstr{\"o}m, L.~Svensson, Y.~Xia, J.~Williams, and {\'A}.~F.
  Garc{\'\i}a-Fern{\'a}ndez, ``Poisson multi-{B}ernoulli mixtures for sets of
  trajectories,'' \emph{arXiv preprint arXiv:1912.08718}, 2019.

\bibitem{blom1988interacting}
H.~A. Blom and Y.~Bar-Shalom, ``The interacting multiple model algorithm for
  systems with {M}arkovian switching coefficients,'' \emph{IEEE Transactions on
  Automatic Control}, vol.~33, no.~8, pp. 780--783, 1988.

\bibitem{garcia2021poisson}
{\'A}.~F. Garc{\'\i}a-Fern{\'a}ndez, J.~L. Williams, L.~Svensson, and Y.~Xia,
  ``A {P}oisson multi-{B}ernoulli mixture filter for coexisting point and
  extended targets,'' \emph{IEEE Transactions on Signal Processing}, vol.~69,
  pp. 2600--2610, 2021.

\bibitem{crouse2016implementing}
D.~F. Crouse, ``On implementing 2{D} rectangular assignment algorithms,''
  \emph{IEEE Transactions on Aerospace and Electronic Systems}, vol.~52, no.~4,
  pp. 1679--1696, 2016.

\bibitem{gospa}
A.~S. Rahmathullah, {\'A}.~F. Garc{\'i}a-Fern{\'a}ndez, and L.~Svensson,
  ``Generalized optimal sub-pattern assignment metric,'' in \emph{Proceedings
  of International Conference on Information Fusion}.\hskip 1em plus 0.5em
  minus 0.4em\relax IEEE, 2017, pp. 1--8.

\bibitem{garcia2020metric}
{\'A}.~F. Garc{\'\i}a-Fern{\'a}ndez, A.~S. Rahmathullah, and L.~Svensson, ``A
  metric on the space of finite sets of trajectories for evaluation of
  multi-target tracking algorithms,'' \emph{IEEE Transactions on Signal
  Processing}, vol.~68, pp. 3917--3928, 2020.

\bibitem{pmbmpoint2}
{\'A}.~F. Garc{\'\i}a-Fern{\'a}ndez, J.~L. Williams, K.~Granstr{\"o}m, and
  L.~Svensson, ``Poisson multi-{B}ernoulli mixture filter: direct derivation
  and implementation,'' \emph{IEEE Transactions on Aerospace and Electronic
  Systems}, vol.~54, no.~4, pp. 1883--1901, 2018.

\end{thebibliography}

\cleardoublepage
{\bfseries \huge Supplementary Materials}

\appendices
\section{Proof of Theorem \ref{thm_multipredicts}}
\label{proof_thm_multipredicts}

We prove Theorem \ref{thm_multipredicts} by induction\footnote{A direct proof using sets integrals can be found in \cite[Appendix A]{xia2020backward}.}. The base case one-step prediction with $\gamma = \eta+1$ has been proved in \cite{garcia2019multiple}. We proceed to show that if \eqref{eq_multistep1} and \eqref{eq_multistep2} hold, then the $(\gamma+1-\eta)$-step predicted multi-trajectory density of $\pi_{\alpha:\eta|k}(\cdot)$ is also of the form \eqref{eq_multistep1} and \eqref{eq_multistep2}. We define ${\bf W}^{\gamma+1}$ as the set of trajectories born at time step $\gamma+1$. Then the one-step predicted multi-trajectory density of $\pi_{\alpha:\gamma|k}({\bf X}_{\alpha:\gamma})$ is 
\begin{multline}
  \label{eq_proof_1}
  \pi_{\alpha:\gamma+1|k}({\bf X}_{\alpha:\gamma+1}) = \pi_{\alpha:\gamma|k}({\bf X}_{\alpha:\gamma})\pi_{\gamma+1:\gamma+1}({\bf W}^{\gamma+1})\\\times \left. \prod_{ \left( t,x^{1:\nu} \right) \in {\bf X}_{\alpha:\gamma+1}^{\gamma}} \right[ \left( 1+p^S\left(x^\nu\right)\left(\delta_{\gamma-t+2}[\nu]-1\right) \right) \\\times\left. \prod_{\ell=\gamma-t+1}^{\nu-1} g \left( x^{\ell+1} | x^\ell\right) p^S\left(x^\ell\right) \right ].
\end{multline}
We further write ${\bf X}_{\alpha:\gamma+1}^{\gamma}  = {\bf Y} \uplus {\bf V}$ where ${\bf Y} \subseteq {\bf X}_{\alpha:\gamma+1}^{\eta}$ is the set of trajectories present at time step $\eta$ and ${\bf V} \subseteq {\bf W}$ is the set of trajectories that appeared after time step $\eta$ (by assumption $\eta < \gamma$). This allows us to separate the product over ${\bf X}_{\alpha:\gamma+1}^{\gamma}$ in \eqref{eq_proof_1} into two products over ${\bf Y}$ and ${\bf V}$, respectively. By plugging \eqref{eq_multistep1} and \eqref{eq_multistep2} into \eqref{eq_proof_1} and combining these two products with the product over ${\bf X}_{\alpha:\gamma}^{\eta}$ in \eqref{eq_multistep1} and the product over ${\bf W}$ in \eqref{eq_multistep2}, respectively, we obtain
\begin{multline}
  \label{eq_proof_2}
  \pi_{\alpha:\gamma+1|k}({\bf X}_{\alpha:\gamma+1}) = \pi_{\alpha:\eta|k}({\bf X}_{\alpha:\eta})\pi_{\eta+1:\gamma+1}({\bf W}\uplus{\bf W}^{\gamma+1})\\\times \left. \prod_{\left(t,x^{1:\nu}\right)\in{\bf X}_{\alpha:\gamma+1}^{\eta}} \right[ \left( 1+p^S\left(x^\nu\right)\left(\delta_{\gamma-t+2}[\nu]-1\right) \right)\\
  \times \left. \prod_{\ell=\eta-t+1}^{\nu-1} g \left( x^{\ell+1} | x^\ell\right) p^S\left(x^\ell\right) \right ],
\end{multline}
\begin{multline}
  \label{eq_proof_3}
  \pi_{\eta+1:\gamma+1}({\bf W}\uplus{\bf W}^{\gamma+1}) =\prod_{\ell=\eta+1}^{\gamma+1} \beta\left(\tau^\ell({\bf W}^{\ell})\right)\\\times \left. \prod_{\left(t,x^{1:\nu}\right)\in{\bf W}}\right[\left( 1+p^S\left(x^\nu\right)\left(\delta_{\gamma-t+2}[\nu]-1\right) \right)\\
  \times \left. \prod_{\ell=1}^{\nu-1} g \left( x^{\ell+1} | x^\ell\right) p^S\left(x^\ell\right)  \right].
\end{multline}
We can then observe that \eqref{eq_proof_2} has the same form as \eqref{eq_multistep1}, and that \eqref{eq_proof_3} has the same form as \eqref{eq_multistep2}. 

This finishes the proof of Theorem \ref{thm_multipredicts}.

\section{Proof of Corollary \ref{corollary_}}
\label{proof_corollary_}

We observe that the only factor in \eqref{eq_multistep1} that depends on time step $\alpha$ is $\pi_{\alpha:\eta|k}({\bf X}_{\alpha:\eta})$. This means that the quotient 
\begin{equation*}
  \frac{\pi_{\alpha:\gamma|k}({\bf X}_{\alpha:\gamma})}{\pi_{\alpha:\eta|k}({\bf X}_{\alpha:\eta})}
\end{equation*} 
does not depend on time step $\alpha$, and therefore it holds that 
\begin{equation}
  \label{eq_proof_21}
  \frac{\pi_{k:\gamma|k}({\bf X}_{k:\gamma})}{\pi_{k:\eta|k}({\bf X}_{k:\eta})} = \frac{\pi_{k+1:\gamma|k}({\bf X}_{k+1:\gamma})}{\pi_{k+1:\eta|k}({\bf X}_{k+1:\eta})}.
\end{equation}
Setting $\eta = k+1$ in \eqref{eq_proof_21} gives
\begin{equation}
  \frac{\pi_{k:\gamma|k}({\bf X}_{k:\gamma})}{\pi_{k:k+1|k}({\bf X}_{k:k+1})} = \frac{\pi_{k+1:\gamma|k}({\bf X}_{k+1:\gamma})}{f_{k+1|k}(\tau^{k+1}({\bf X}_{k+1:k+1}))},
\end{equation}
which can be rearranged to obtain \eqref{eq_corollary}.

This finishes the proof of Corollary \ref{corollary_}.

\section{Proof of Theorem \ref{thm_smoothing}}
\label{proof_thm_smoothing}

Let ${\bf X}$ and ${\bf Y}$ be the set of trajectories in the time interval $k:K$ and $k+1:K$, respectively. Applying the total probability theorem and Bayes' rule gives
\begin{align}
  \pi_{k:K|K}({\bf X}) &= \int \pi_{k:K|K}({\bf X},{\bf Y})\delta {\bf Y},\nonumber \\
  &= \int \pi_{k+1:K|K}({\bf Y})\pi_{k:K|K}({\bf X}|{\bf Y})\delta {\bf Y}
\end{align}
where $\pi_{k:K|K}({\bf X}|{\bf Y})$ defines the backward transition density from ${\bf Y}$ to ${\bf X}$ conditioned on the sequence of sets of measurements up to and including time step $K$. Assume that ${\bf X}$ is independent of measurements $({\bf z}_{k+1},\dots,{\bf z}_K)$ that are in the future given ${\bf Y}$:
\begin{equation}
  \label{eq_appendix_1}
  \pi_{k:K|K}({\bf X}|{\bf Y}) = \pi_{k:K|k}({\bf X}|{\bf Y}).
\end{equation}
Then we have 
\begin{align}
  \pi_{k:K|K}({\bf X}) &= \int \pi_{k+1:K|K}({\bf Y})\pi_{k:K|k}({\bf X}|{\bf Y})\delta {\bf Y}\nonumber\\
  &= \pi_{k:K|k}({\bf X})\int \pi_{k+1:K|K}({\bf Y})\frac{\pi_{k:K|k}({\bf X}|{\bf Y})}{\pi_{k:K|k}({\bf X})}\delta {\bf Y}.
\end{align}
Bayes' rule then yields
\begin{equation}
  \label{eq_proof_31}
  \pi_{k:K|K}({\bf X}) = \pi_{k:K|k}({\bf X})\int \pi_{k+1:K|K}({\bf Y})\frac{\pi_{k+1:K|k}({\bf Y}|{\bf X})}{\pi_{k+1:K|k}({\bf Y})}\delta {\bf Y}
\end{equation}
where $\pi_{k+1:K|k}({\bf Y}|{\bf X})$ defines the transition density from ${\bf X}$ to ${\bf Y}$, which is a multi-trajectory Dirac delta. The integral over ${\bf Y}$ in \eqref{eq_proof_31} can then be cancelled out by applying the prediction equation for sets of trajectories \cite[Eq. (8)]{garcia2019multiple}, which gives us
\begin{equation}
  \label{eq_proof_32}
  \pi_{k:K|K}({\bf X}) = \frac{\pi_{k:K|k}({\bf X})\pi_{k+1:K|K}({\bf X}_{k+1:K})}{\pi_{k+1:K|k}({\bf X}_{k+1:K})}.
\end{equation}
Applying Corollary \ref{corollary_} on \eqref{eq_proof_32} yields
\begin{equation}
  \label{eq_proof_33}
  \frac{\pi_{k:K|k}({\bf X})}{\pi_{k+1:K|k}({\bf X}_{k+1:K})} = \frac{\pi_{k:k+1|k}({\bf X}_{k:k+1})}{f_{k+1|k}\left(\tau^{k+1}({\bf X}_{k+1:k+1})\right)}.
\end{equation}
The proof is finished by plugging \eqref{eq_proof_33} into \eqref{eq_proof_32}.

\section{Proof of Lemma \ref{lemma_bs}}
\label{proof_lemma_bs}

We first rewrite $\pi_{k:K|K}({\bf X}|{\bf Y})$ using \eqref{eq_appendix_1} and Bayes' rule:
\begin{equation}
  \pi_{k:K|K}({\bf X}|{\bf Y}) = \frac{\pi_{k:K|k}({\bf X})\pi_{k+1:K|k}({\bf Y}|{\bf X})}{\pi_{k+1:K|k}({\bf Y})}
\end{equation}
where $\pi_{k+1:K|k}({\bf Y}|{\bf X})$ defines the transition density from ${\bf X}$ to ${\bf Y}$, which is a multi-trajectory Dirac delta. We then apply Corollary \ref{corollary_}, which gives
\begin{equation}
  \pi_{k:K|K}({\bf X}|{\bf Y}) = \frac{\pi_{k:k+1|k}({\bf X}_{k:k+1})\delta_{\bf Y}({\bf X}_{k+1:K})}{f_{k+1|k}\left(\tau^{k+1}({\bf Y}_{k+1:k+1})\right)}.
\end{equation}
As $f_{k+1|k}\left(\tau^{k+1}({\bf Y}_{k+1:k+1})\right)$ does not depend on ${\bf X}$, we can further express $\pi_{k:K|k}({\bf X}|{\bf Y})$ as \eqref{eq_lemma_bs}.

This finishes the proof of Lemma \ref{lemma_bs}.

\section{Proof of Theorem \ref{thm_pmb}}
\label{proof_thm_pmb}

We first give the explicit expression of the multi-trajectory density $\pi_{k:k+1|k}({\bf X}_{k:k+1})$. Given the PMB filtering density at time step $k$ (cf. \eqref{eq_pmb_whole}) and the multi-trajectory dynamic model described in Section \ref{dynamic_model} with Poisson birth density \eqref{eq_poisson_birth}, the predicted multi-trajectory density $\pi_{k:k+1|k}({\bf X}_{k:k+1})$ is a PMB \cite[Lemma 4]{garcia2020trajectory}
\begin{multline}
  \label{eq_trajectory_pmb}
  \pi_{k:k+1|k}({\bf X}_{k:k+1}) =\\ \sum_{\uplus_{j=1}^{n_{k:k+1|k}}{\bf X}^j\uplus {\bf P} = {\bf X}_{k:k+1}}\pi^p_{k:k+1|k}({\bf P})\prod_{i=1}^{n_{k:k+1|k}}\left[\pi^i_{k:k+1|k}\left({\bf X}^i\right)\right]
\end{multline}
with $n_{k:k+1|k} = n_{k|k}$, where $\pi^p_{k:k+1|k}(\cdot)$ is of the form \eqref{eq_trajectory_ppp} with intensity
\begin{multline}
  \label{eq_trajectory_ppp_proof}
  \lambda^u_{k:k+1|k}\left(t,x^{1:\nu}\right) = \delta_{k+1}[t]\delta_{1}[\nu]\lambda^B_{k+1}\left(x^1\right) \\+ \left\langle \lambda^u_{k|k}\left(y^1\right)\delta_{k}[t^\prime]\delta_{1}[\nu^\prime],g^{k+1}\left(t,x^{1:\nu}|t^\prime,y^{1:\nu^\prime}\right) \right\rangle,
\end{multline}
and $\pi^i_{k:k+1|k}({\bf X}^i)$ is of the form \eqref{eq_Bernoulli}, parameterised by
\begin{align}
  r^i_{k:k+1|k} &= r^i_{k|k},\label{eq_proof_e_ber1}\\
  p^i_{k:k+1|k}(X) &= \left\langle p^i_{k|k}\left(y^1\right)\delta_{k}[t^\prime]\delta_{1}[\nu^\prime],g^{k+1}\left(X|t^\prime,y^{1:\nu^\prime}\right) \right\rangle.\label{eq_proof_e_ber2}
\end{align}
where the single trajectory transition density $g^{k+1}(\cdot|\cdot)$ is given by \eqref{eq_single_trajectory_transition}.

Next, we elaborate on why the multi-trajectory Dirac delta $\delta_{{\bf Y}}({\bf X}_{k+1:K})$ can be seen as a standard multi-object measurement model \cite{mahler2007statistical} with characteristics specified in Section \ref{sec_pmbm_recursion}. We note that the multi-trajectory Dirac delta $\delta_{{\bf Y}}({\bf X}_{k+1:K})$ can be written as a trajectory MB where each trajectory Bernoulli component is parameterised by probability of existence one and a Dirac delta single-trajectory density, i.e.,
\begin{equation}
  \label{eq_tra_dirac_delta}
  \delta_{{\bf Y}}({\bf X}_{k+1:K}) = \sum_{\uplus_{j=1}^{n_{k+1:K}}{\bf X}^j = {\bf X}_{k+1:K}} \prod_{i=1}^{n_{k+1:K}}\delta_{\{Y^i\}}\left({\bf X}^i\right)
\end{equation}
where for trajectories that did not exist in the time interval $k+1:K$, they are implicitly represented by trajectory Bernoulli components with zero probability of existence. Therefore, the multi-trajectory Dirac delta $\delta_{{\bf Y}}({\bf X}_{k+1:K})$ can be understood as a standard measurement model for sets of trajectories \cite{garcia2020trajectory}: 
\begin{itemize}
  \item Each trajectory $X = (t,x^{1:\nu}) \in {\bf X}_{k:K}$ is detected with probability 
  \begin{equation}
    p^D(X) = \begin{cases}
      0, & t = k~\text{and}~\nu = 1\\
      1, & \text{otherwise}
    \end{cases}
  \end{equation}
  and if detected, it generates a measurement $Y$ with density $\delta_{Y}(X_{k+1:K})$.
  \item Poisson clutter intensity $\lambda^C(\cdot) = 0$.
\end{itemize}
Note that if a trajectory $X$ in the time interval $k:K$ did not exist in the time interval $k+1:K$, then it must have start time $t=k$ and length $\nu=1$.

Having established the analogy between the backward kernel \eqref{eq_lemma_bs} for PMB filtering densities and the trajectory PMB update, the explicit expression of the trajectory PMBM backward kernel \eqref{eq_pmbm} can be straightforward obtained by plugging the trajectory PMB prior \eqref{eq_trajectory_pmb} and measurement model \eqref{eq_tra_dirac_delta} into the trajectory PMB update equations \cite[Lemma 5]{garcia2020trajectory}. One thing to be noted is that the set of trajectories appeared after time step $k+1$ remains unaltered. This will be further elaborated.

We write ${\bf X} = {\bf X}^\prime \uplus {\bf D}^\prime$ as the disjoint union of the set ${\bf X}^\prime$ of trajectories present at time step $k$ or $k+1$ and the set ${\bf D}^\prime$ of trajectories of objects that appeared after time step $k+1$. We also write ${\bf Y} = {\bf Y}^\prime \uplus {\bf D}$ as the disjoint union of the set ${\bf Y}^\prime$ of trajectories of objects that existed at time step $k+1$ and the set ${\bf D}$ of trajectories of objects that appeared after time step $k+1$. 
It holds that ${\bf X}_{k+1:K}= {\bf X}_{k+1:K}^\prime\uplus {\bf D}^\prime$ by construction, and that the density of ${\bf X}_{k+1:K}$ conditioned on ${\bf Y}$ is zero unless ${\bf D} = {\bf D}^\prime$. Then we have
\begin{equation}
  \label{eq_proof_e1}
  \delta_{{\bf Y}^\prime \uplus {\bf D}}\left({\bf X}^\prime_{k+1:K} \uplus {\bf D}^\prime\right) = \delta_{{\bf Y}^\prime}\left({\bf X}_{k+1:K}^\prime\right)\delta_{{\bf D}}\left({\bf D}^\prime\right),
\end{equation}
and the backward kernel \eqref{eq_lemma_bs} becomes
\begin{multline}
  \label{eq_proof_e2}
  \pi_{k:K|K}\left({\bf X}^\prime \uplus {\bf D}^\prime|{\bf Y}^\prime \uplus {\bf D}\right) \\ \propto \pi_{k:k+1|k}\left({\bf X}^\prime_{k:k+1}\right)\delta_{{\bf Y}^\prime}\left({\bf X}_{k+1:K}^\prime\right)\delta_{{\bf D}}\left({\bf D}^\prime\right)
\end{multline}
where $\delta_{{\bf D}}\left({\bf D}^\prime\right)$ can be written as a trajectory MB similar to \eqref{eq_tra_dirac_delta}. This explains the last part of Theorem \ref{thm_pmb} and finishes the proof of Theorem \ref{thm_pmb}.

\section{}
\label{statistics}
For the particle representation of the multi-trajectory density \eqref{eq_particle}, the cardinality distribution of the set ${\bf X}_{\alpha:\gamma}$ of trajectories is given by
\begin{equation}
  \text{Pr}(|{\bf X}_{\alpha:\gamma}|=n) = \sum_{i=1}^{T}w^{(i)}\delta_n\left[\left|{\bf X}^{(i)}\right|\right],
\end{equation}
the cardinality distribution of the trajectories born at time step $k$ with $\alpha \leq k \leq \gamma$ is given by
\begin{equation}
  \text{Pr}(n~\text{births at time}~k) = \sum_{i=1}^{T}w^{(i)}\delta_{n}\left[\sum_{(t,x^{1:\nu})\in{\bf X}^{(i)}}\delta_k[t]\right],
\end{equation}
and the cardinality distribution of the trajectories that die at time step $k$ with $\alpha \leq k \leq \gamma-1$ is given by  
\begin{align}
  &\text{Pr}(n~\text{deaths at time}~k)\nonumber\\
  &~~~= \sum_{i=1}^{T}w^{(i)}\delta_{n}\left[\sum_{(t,x^{1:\nu})\in{\bf X}^{(i)}}\delta_k[t+\nu-1]\right].
\end{align}

\section{}
\label{pseudo_code}
The pseudocode of linear-Gaussian backward simulation for sets of trajectories with PMB filtering densities along with an efficient estimator is given in Algorithm \ref{alg1}. The proposed estimator reports the particle with the highest likelihood accumulated over time from $T$ particles (sets of trajectories) with equal weight $1/T$.

\begin{algorithm}[!t]
  \footnotesize
  \caption{Backward simulation for sets of trajectories with PMB filtering densities}
    \label{alg1}
    \begin{algorithmic}[1]
      \REQUIRE $T$, $\left\{\lambda^B_k(\cdot), \lambda^u_{k|k}(\cdot)\right\}_{k=1}^{K-1}$, $\left\{\left\{r^j_{k|k},p^j_{k|k}(\cdot)\right\}_{j=1}^{n_{k|k}}\right\}_{k=1}^K$.
      \ENSURE ${\bf X}_{1:K}$
      \FOR{$\iota = 1,\dots,T$}
      \STATE ${\bf y}_{K|K} = \emptyset$; 
      \STATE $c_\iota = 0$;
      \FOR{$i = 1,\dots,n_{K|K}$}
      \STATE $u \sim \text{Uniform}[0,1]$;
      \IF{$u \leq r_{K|K}^i$}
      \STATE ${\bf y}_{K|K} = {\bf y}_{K|K} \cup \left\{x^i_{K|K}\right\}$; 
      \ENDIF
      \ENDFOR
      \STATE ${\bf Y}_{K:K} = \left\{Y = \left(1,y^1\right): y^1 \in {\bf y}_{K|K}\right\}$
      \FOR{$k = K-1,\dots,1$}
      \STATE Separate ${\bf Y}_{k+1:K}$ into $\left\{Y^j\right\}_{j=1}^m$ and $\left\{Y^j\right\}_{j=m+1}^{n_{k+1:K}}$ as described in Theorem \ref{thm_pmb};
      \FOR{$i = 1,\dots,n_{k|k}$}
      \FOR{$j = 1,\dots,m$}
      \IF{$\text{SMD}(y_j^1;x_{k|k}^i,P_{k|k}^i) < \Gamma_g$}
      \STATE Compute $W_1^{(j,i)}$ using \eqref{eq_weight_matrix_entry};
      \ELSE
      \STATE $W_1^{(j,i)} = 0$;
      \ENDIF
      \ENDFOR
      \ENDFOR
      \STATE Compute $W_2$ \eqref{eq_weight_matrix_entry2};
      \STATE $C = -\log \begin{bmatrix}
        W_1 & W_2
      \end{bmatrix}$;
      \STATE Run Murty's algorithm on $C$ to obtain the $M$ best global hypotheses with highest weight $a^*_{k:K|K} = \{a_1,\dots,a_M\}$;
      \STATE Compute $[\hat{w}_{a_1},\dots,\hat{w}_{a_M}]$ using \eqref{eq_glo_hyo_weight} and normalise them to obtain $[w_{a_1},\dots,w_{a_M}]$;
      \STATE $a \sim \text{Categorical}([w_{a_1},\dots,w_{a_M}])$; 
      \STATE $c_\iota = c_\iota + \log(\hat{w}_a)$;
      \ENDFOR
      \STATE ${\bf Y}_{k:K} = \left\{Y^j\right\}_{j=m+1}^{n_{k+1:K}}$;
      \FOR{$i = 1,\dots,n_{k|k}$}
      \IF{$a^i = 1$}
      \STATE $u \sim \text{Uniform}[0,1]$;
      \IF{$u \leq r^{i,1}_{k:K|K}$ \eqref{eq_mis_r}}
      \STATE Sample $Y \sim p^{i,1}_{k:K|K}(\cdot)$ using \eqref{eq_mis};
      \STATE ${\bf Y}_{k:K} = {\bf Y}_{k:K} \cup \{Y\}$;
      \ENDIF
      \ELSE
      \STATE Sample $Y \sim p^{i,a^i}_{k:K|K}(\cdot)$ using \eqref{eq_Gaussian_cov2};
      \STATE ${\bf Y}_{k:K} = {\bf Y}_{k:K} \cup \{Y\}$;
      \ENDIF
      \ENDFOR
      \FOR{$i = n_{k|k}+1,\dots,n_{k|k} + m$}
      \IF{$a^i = 2$}
      \STATE Sample $Y \sim p^{i,2}_{k:K|K}(\cdot)$ using \eqref{eq_lemma_ppp};
      \STATE ${\bf Y}_{k:K} = {\bf Y}_{k:K} \cup \{Y\}$;
      \ENDIF
      \ENDFOR
      \STATE ${\bf X}^{(\iota)}_{1:K} = {\bf Y}_{1:K}$;
      \ENDFOR
      \STATE $\iota^* = \arg\max_{\iota}c_\iota$; 
      \STATE ${\bf X}_{1:K} = {\bf X}^{(\iota^*)}_{1:K}$;
    \end{algorithmic}
  \end{algorithm}

\section{}
\label{appendix_example_expression}

In this appendix, we derive the explicit expression of multi-trajectory smoothing density $\pi_{K-1:K|K}({\bf X}_{K-1:K})$ for the one-dimensional example in Section \ref{sec_example}. 

We write ${\bf X}_{K-1:K} = {\bf Y} \uplus {\bf V} \uplus {\bf B}$ where the set ${\bf Y}$ of trajectories present at both time step $K-1$ and $K$, the set ${\bf V}$ of trajectories only present at time step $K-1$, and the set ${\bf B}$ of trajectories only present at time step $K$. We also write the multi-object filtering density $f_{k|k}({\bf x})$ at time step $k$ as $\delta_{{\bf y}_k}({\bf x})$ where ${\bf y }_k = \{y_k^1,y_k^2\}$ and $k\in\{1,\dots,K\}$. The predicted multi-trajectory density of $f_{K-1|K-1}(\cdot)$ for each of the events can then be expressed using Theorem \ref{thm_multipredicts}:
\begin{itemize}
  \item The two objects existed at both time step $K-1$ and $K$,
  \begin{align*}
    &\pi_{K-1:K|K-1}({\bf Y}\uplus{\bf B})\\
    &= \delta_{{\bf y}_{K-1}}\left(\tau^{K-1}({\bf Y})\right)e^{-0.1}\prod_{\left(K,x^1\right)\in{\bf B}}\lambda^B\left(x^1\right)\\
    &~~~\times\left(p^S\right)^2\prod_{\left(K-1,x^{1:2}\right)\in{\bf Y}}{\cal N}\left(x^2;Fx^1,Q\right).
  \end{align*}
  \item Object with state $y_{K-1}^1$ existed at time step $K$ and object with state $y_{K-1}^2$ died at time step $K-1$,
  \begin{align*}
    &\pi_{K-1:K|K-1}({\bf Y}\uplus{\bf V}\uplus{\bf B})\\ 
    &= \delta_{{\bf y}_{K-1}}\left(\tau^{K-1}({\bf Y}\uplus{\bf V})\right)e^{-0.1}\prod_{\left(K,x^1\right)\in{\bf B}}\lambda^B\left(x^1\right)\\
    &~~~\times\left(1-p^S\right)p^S\prod_{\left(K-1,x^{1:2}\right)\in{\bf Y}}{\cal N}\left(x^2;Fy_{K-1}^1,Q\right).
  \end{align*}
  \item Object with state $y_{K-1}^2$ existed at time step $K$ and object with state $y_{K-1}^1$ died at time step $K-1$,
  \begin{align*}
    &\pi_{K-1:K|K-1}({\bf Y}\uplus{\bf V}\uplus{\bf B})\\ 
    &= \delta_{{\bf y}_{K-1}}\left(\tau^{K-1}({\bf Y}\uplus{\bf V})\right)e^{-0.1}\prod_{\left(K,x^1\right)\in{\bf B}}\lambda^B\left(x^1\right)\\
    &~~~\times\left(1-p^S\right)p^S\prod_{\left(K-1,x^{1:2}\right)\in{\bf Y}}{\cal N}\left(x^2;Fy_{K-1}^2,Q\right).
  \end{align*}
  \item The two objects died at time step $K-1$,
  \begin{multline*}
    \pi_{K-1:K|K-1}({\bf V}\uplus{\bf B})
    =\\ \delta_{{\bf y}_{K-1}}\left(\tau^{K-1}({\bf V})\right)e^{-0.1}\prod_{\left(K,x^1\right)\in{\bf B}}\lambda^B\left(x^1\right)\left(1-p^S\right)^2.
  \end{multline*}
\end{itemize}
When taking the product of the two multi-trajectory densities $\pi_{K-1:K|K}({\bf X}_{K-1:K})$ and $\pi_{K:K|K}({\bf X}_{K|K})$, a multi-trajectory Dirac delta with two components, each of the first three events has two different cases due to the unknown correspondence between $\tau^{K-1}({\bf Y})$ and ${\bf y}_K$. Therefore, the multi-trajectory density $\pi_{K-1:K|K}({\bf X}_{K-1:K})$ has a mixture representation of $7$ components.

\ifCLASSOPTIONcaptionsoff
  \newpage
\fi

\end{document}